\documentclass[acmsmall,screen,nonacm]{acmart}

\usepackage{pgf-pie}

\usepackage{listings}
\definecolor{palered}{rgb}{1,0.8,0.8}
\definecolor{mykeywords}{RGB}{0,0,255} % Blue color
\definecolor{mov}{RGB}{242,134,197} % purple color
\definecolor{del}{RGB}{245,117,117} % red color
\definecolor{upd}{RGB}{219,175,38} % yellow color
\definecolor{ins}{RGB}{87,208,61}

\definecolor{myorange}{RGB}{250,153,45}
\definecolor{myblue}{RGB}{41,134,204}
\definecolor{mygreen}{RGB}{120,190,84}
\theoremstyle{definition}
\newtheorem{theorem}{Theorem}

\usepackage[english]{babel}
\usepackage{amsthm}
\usepackage{mathpartir}
\usepackage{bussproofs}
\usepackage{pdfpages}
\usepackage{caption}
\usepackage{multirow}
\usepackage{listings}
\usepackage{graphicx}
\usepackage{subfig}
\usepackage{lipsum}
\usepackage[linesnumbered,ruled,vlined,algo2e]{algorithm2e}
\SetAlFnt{\small}
\usepackage{algpseudocode}
\usepackage{amsmath}

\usepackage{algpseudocode}
\usepackage{amsmath,accents}
\usepackage{enumitem}
\usepackage{cases}

\usepackage{tikz}
\usetikzlibrary{trees}
\usepackage{forest}
\usepackage{wrapfig}

\newtheorem{lemma}[theorem]{Lemma}
\newtheorem{definition}{Definition}[section]
\newtheorem{remark}{Remark}

\usepackage{diagbox}

\AtBeginEnvironment{prooftree}{\small}

\begin{document}

\title{SAT-DIFF: A Tree Diffing Framework Using SAT Solving}

% \affiliation{%

\author{Chuqin Geng}
\email{chuqin.geng@mail.mcgill.ca}
\affiliation{%
  \institution{McGill University}
  \country{Canada}
}

\author{Haolin Ye}
\email{haolin.ye@mail.mcgill.ca}
\affiliation{%
  \institution{McGill University}
  \country{Canada}
}

\author{Yihan Zhang}
\email{yihan.zhang2@mail.mcgill.ca}
\affiliation{%
  \institution{McGill University}
  \country{Canada}
}

\author{Brigitte Pientka}
\email{bpientka@cs.mcgill.ca}
\affiliation{%
  \institution{McGill University}
  \country{Canada}
}

\author{Xujie Si}
\email{six@cs.toronto.edu}
\affiliation{%
  \institution{University of Toronto}
  \country{Canada}
}

\newcommand{\allen}[1]{{\color{red} [\textbf{Allen}: #1]}}
\newcommand{\xs}[1]{{\color{purple} [\textbf{Xujie}: #1]}}
\newcommand{\jl}[1]{{\color{magenta} [\textbf{Jaylene}: #1]}}
\newcommand{\haolin}[1]{{\color{brown} [\textbf{Haolin}: #1]}}

\newcommand{\maxsat}{MaxSAT}
%{\Large M}AX{\Large S}AT
\newcommand{\satdiff}{\textsc{SatDiff}}  
\newcommand{\gumtree}{\textit{Gumtree}}
\newcommand{\truediff}{\textit{truediff}}

\newcommand{\dcon}[3]{\texttt{dcon}(#1, e_{#2},#3)}
\newcommand{\con}[3]{\texttt{con}(#1,e_{#2},#3)}  
\newcommand{\del}[3]{\texttt{del}(#1, e_{#2},#3)}  

\newcommand{\Upd}[3]{\texttt{Upd}(#1, #2,#3)}
\newcommand{\Mov}[5]{\texttt{Mov}(#1, #2,e_{#3},#4,e_{#5})}
\newcommand{\Ins}[3]
{\texttt{Ins}(#1,e_{#2},#3)}  
\newcommand{\Del}[3]{\texttt{Del}(#1, e_{#2},#3)}  
\newcommand{\Swp}[6]{\texttt{Swp}(#1, #2, e_{#3},#4,#5,e_{#6})}

\newcommand{\fe}[2]{(#1,e_{#2})}
\newcommand{\halfe}[2]{(#1,e_{#2})}

\newcommand{\evar}[3]{x_{{#1,e_{#2},#3}}}
\newcommand{\negevar}[3]{\neg x_{#1,e_{#2},#3}}

\newcommand{\mvar}[2]{x_{#1 \leftrightarrow #2}}
\newcommand{\negmvar}[2]{\neg x_{#1 \leftrightarrow #2}}
\newcommand{\nd}[1]{${#1}$}

%%
%% By default, the full list of authors will be used on the page


% While traditional methods like Unix {\fontfamily{cmss}\selectfont diff} are effective at the text granularity level, their edit scripts often lack clarity and fail to capture programmers' intentions. To address this issue, a line of research has emerged focusing on computing edit scripts at the abstract syntax tree (AST) granularity.
% However, existing approaches often rely on heuristics to identify matches across different versions of a program, resulting in suboptimal edit scripts. 

% A significant challenge in software evolution lies in identifying code changes or edit scripts applied to source code files. Directly computing differences over abstract syntax trees (ASTs) offers improved clarity and better alignment with the programmer's intention compared to classical Unix {\fontfamily{cmss}\selectfont diff}. 

\begin{abstract}

Computing differences between tree-structured data is a critical but challenging problem in software analysis. In this paper, we propose a novel tree diffing approach called {\satdiff}, which reformulates the structural diffing problem into a {\maxsat} problem. By encoding the necessary transformations from the source tree to the target tree, {\satdiff} generates correct, minimal, and type safe low-level edit scripts with formal guarantees. We then synthesize concise high-level edit scripts by effectively merging low-level edits in the appropriate topological order. Our empirical results demonstrate that {\satdiff} outperforms existing heuristic-based approaches by a significant margin in terms of conciseness while maintaining a reasonable runtime.

% Tree-structured data differencing is a challenging problem in software evolution analysis. Our paper introduces SATDiff, a novel tree diffing approach that transforms the structural differencing problem into a MaxSAT problem. By using well-designed constraints, the MaxSAT solver produces correct, minimal, and type-safe low-level edit scripts. We synthesize concise high-level edit scripts by efficiently merging low-level edits in the appropriate topological order. Our empirical results show that SATDiff outperforms existing heuristic-based approaches significantly while maintaining comparable runtime performance.

% A fundamental problem of software evolution analysis is to identify code changes or edit scripts made to source code files. While traditional approaches like  Unix {\fontfamily{cmss}\selectfont diff} can efficiently solve this problem at the text granularity, their edit scripts are often obscure and fail to reflect the intentions of programmers. To address this issue, there is a line of work that aims to compute edit scripts at the abstract syntax tree granularity. However, these methods often rely on heuristics to find matches across different versions of a program, resulting in less optimal edit scripts. In this paper, we propose a novel approach called {\satdiff} that encodes the code differencing problem to a  {\Large M}AX{\Large S}AT problem, allowing for the finding of optimal edit scripts. Leveraging state-of-the-art solvers, {\satdiff} can find optimal solutions within a reasonable amount of time. Our empirical results suggest that...

\end{abstract}

\received{20 February 2007}
\received[revised]{12 March 2009}
\received[accepted]{5 June 2009}

   \maketitle

% \tableofcontents

\section{Introduction}
\label{sec:intro}

Software systems need to undergo evolution to fulfill their intended objectives, thereby contributing to the emergence of the field known as \emph{software evolution analysis} \cite{SoftwareEvolution}.
One of the most fundamental problems in software evolution analysis is \emph{source code differencing}, which has wide applications in version control systems, program analysis, and debugging. The problem entails computing differences between two source code files and determining the transformation from one to the other through a series of edit actions, commonly referred to as \textit{edit scripts}.  In typical software development environments, edit scripts are computed between different versions of the same code to understand changes applied to the original code.

Viewing code as text is a straightforward approach for computing edit scripts. For instance, when given two versions of a source code file, the Unix {\fontfamily{cmss}\selectfont diff} tool utilizes the Myers algorithm \cite{DBLP:journals/algorithmica/Meyers86} at the text level of granularity to generate an edit script that identifies added or deleted lines. However, edit scripts generated by tools like Unix {\fontfamily{cmss}\selectfont diff} and other text-level methods can often be challenging for users to interpret. These methods do not take the structure of the code into account, thereby necessitating further analysis in order to understand behind code changes.

To address this issue, a line of work \cite{ChangeDistilling,gumtree,JSync,XDiff,MTDIFF} has emerged that computes edit scripts at the abstract syntax tree (AST) level. This is also known as the \emph{tree/structural diffing} problem. By incorporating the structural information of the code, tree diffing edit scripts support high-level edit actions such as \textit{update}, \textit{move}, \textit{insert}, and \textit{delete},
which can significantly enhance readability for end users. One widely adopted design choice for computing tree diffing is the algorithm proposed by \cite{ChawatheRGW96}. This method involves addressing two sub-problems: 1) identifying a suitable matching between the source and target trees, and 2) deducing an edit script based on the established matching. It is worth noting that finding the shortest or minimum tree diffing edit script has been proven to be NP-hard \cite{DBLP:journals/tcs/Bille05}, and most existing methods, such as \textit{Gumtree} \cite{gumtree}, follow the approach of Chawathe et al., often failing to produce minimum edit scripts.

% Most existing methods such as \textit{Gumtree} \cite{gumtree} rely on heuristics to search for good matching and deduce plausible edit scripts, thus usually failing to yield optimal solutions.

% However, it is worth noting that finding the optimal or minimum tree-differencing edit script has been proven to be NP-hard \cite{DBLP:journals/tcs/Bille05}. Consequently, most existing methods rely on heuristics to search for good matching and deduce plausible edit scripts, thus usually failing to yield optimal solutions.

In this paper, our objective is to compute minimum edit scripts between two code trees while providing a formal guarantee of minimality. Our approach, called {\satdiff}, can be broken down into three distinct phases. First, in contrast to previous work, we reframe the tree diffing problem as a maximum satisfiability ({\maxsat}) problem \cite{stutzle2001review} and leverage state-of-the-art solvers to search for the correct minimum edits. Our intuition behind {\satdiff} is rather straightforward: we employ a solver designed for NP (nondeterministic polynomial)  problems to tackle the matching problem's intrinsic NP complexity. While solving NP problems such as tree diffing can raise runtime concerns, recent advancements in solvers, along with our optimizations of the constraints encoding, have made it feasible to obtain results in a timely manner. 

Second, we decode the solver's solution into low-level edit actions, which primarily specify operations on the nodes and edges of trees. We also show that these low-level edit actions can form low-level edit scripts when arranged in an appropriate order. Furthermore, our encoding constraints provide correctness and minimality guarantees for low-level edit scripts. It is also worth highlighting that edit scripts generated by $\satdiff$ can achieve the same level of type safety as proposed in \textit{truediff} \cite{truediff}. Specifically, all intermediate trees resulting from each edit action remain well-typed, even if they may contain holes.

% Although we mainly adopt common high-level edit actions introduced in previous work, we also make certain modifications to improve the readability. For example, \textit{Gumtree}'s \textit{delete} action is limited to deleting a single node, requiring multiple \textit{delete} actions to delete a subtree. In contrast, our \texttt{delete} action enables the direct deletion of a subtree, significantly simplifying the edit script's size.

In the final step, we synthesize high-level edit scripts from low-level edit actions, presenting them in a more intuitive and user-friendly format for end users. Specifically, our low-level edit actions can be combined to create high-level edit actions, such as \textit{update}, \textit{move}, \textit{insert}, and \textit{delete}. This process is accomplished using an efficient algorithm that combines low-level edit actions into coherent high-level edit actions while considering the appropriate topological order. Our experimental results also indicate that our high-level edit scripts outperform existing approaches such as \textit{truediff} and \textit{Gumtree} in terms of conciseness while maintaining a reasonable runtime. In summary, we make the following contributions:
\begin{itemize}
    \item We present a novel approach --- {\satdiff} for tree diffing by formulating the problem as a {\maxsat} problem. We encode tree edits into {\maxsat} constraints, where hard constraints ensure the correctness of edits, and soft constraints guarantee the minimality of the edits.  
    \item We interpret {\maxsat} solutions into low-level edit actions and subsequently synthesize high-level edit scripts based on the dependency graph of low-level edit actions.    
    % defined edit actions... on text-level instead of from tree-level solution.
    \item We present rigorous semantics of our low-level and high-level edit scripts, which transform from the source tree to the target tree in a type-safe manner.

    \item We evaluate the conciseness and runtime performance of {\satdiff}, comparing it to state-of-art approaches such as \textit{Gumtree} and \textit{truediff}. Additionally, we conduct an ablation study to demonstrate the effectiveness of our encoding constraints and a case study to understand discrepancies between {\satdiff} and \textit{Gumtree}.

    % Our empirical results indicate that {\satdiff} outperforms existing approaches, such as \textit{Gumtree} and \textit{truediff}, in terms of conciseness on real-world benchmarks. Meanwhile, {\satdiff} exhibits reasonable running times. 
\end{itemize}

\vspace{-10pt}

\section{Preliminary -- The {\maxsat} problem}

\label{sec:back}

In order to discuss the maximum satisfiability (\maxsat) problem, we must first introduce the Boolean satisfiability (SAT) problem. The SAT problem aims to determine whether a given propositional formula is satisfiable and the problem is known to be NP-complete \cite{DBLP:conf/stoc/Cook71}. It can be presented in a standard form called conjunctive normal form (CNF), where it is expressed as a conjunction of clauses, each consisting of a disjunction of literals.

To be more precise, suppose we have a set of $n$ propositional Boolean variables for 
a CNF $\Phi$, denoted as 
$ vars(\Phi) =  \{x_1,x_2,...,x_n\}$. A truth assignment $\pi$ for $\Phi$ assigns each propositional variable  $ x_i \in vars(\Phi)$ to a truth value
(true or false). We say that a \emph{literal} $l$ is either a variable $x$  or its negation $\neg x$.  For a literal $l$, $\pi$ makes $l$ true ($\pi \models l$) if $l$ is the variable $x$ and $\pi(x) = true$ or if $l$ is $\neg x$ and $\pi(x) = false$. A \emph{clause} $C$ that consists of $n$ literals can be expressed as
\[
C = \bigvee_{i=1}^{n} l_{i} 
\]

We say a clause $C$ is satisfied by $\pi$ if $\pi$ makes at least one of the literals $l_i$ in $C$ true, denoted as $\pi \models C$. Then, a CNF $\Phi$ which contains $m$ clauses $\{C_1,C_2,...,C_m\}$ can be represented as
\[
\Phi = \bigwedge_{j=1}^{m}C_j = \bigwedge_{j=1}^{m}\bigvee_{i=1}^{n_j} l_{ij}
\]

It is trivial to show that $\pi$ satisfies $\Phi$ ($\pi \models \Phi$) if and only if all $\{C_1,C_2,...,C_m\}$ are satisfied by $\pi$ simultaneously. In this case, we say $\pi$ is a model of $\Phi$. The goal of solving the {\Large S}AT problem is to return a model $\pi$ if $\Phi$ can be satisfied, or a proof if $\Phi$ is not satisfiable.

In certain scenarios, it may be impossible to find an assignment $\pi$ that satisfies all clauses, thereby giving rise to the {\maxsat} problem. A {\maxsat} CNF $\Phi$ can be partitioned into hard and soft clauses: 
\[
\Phi = hard(\Phi) \bigwedge soft(\Phi)
\]

where hard clauses $hard(\Phi)$ must be satisfied, while soft clauses $soft(\Phi)$ can be falsified at a certain cost. More specifically, the cost of a soft clause $C$ is a positive weight denoted as $wt_C$. 

% For example, in this CNF $\Phi = (x \vee \neg y \vee z) \wedge (\neg x \vee z) \wedge  ( \neg x)_3 \wedge (\neg y)_{10} $, the first two clauses are hard and the third and fourth
% clauses are soft with weights of 3 and 10 respectively.

A \emph{feasible} solution of $\Phi$ is a truth assignment $\pi$ to $vars(\Phi)$( the variables of $\Phi$), that satisfies the hard clauses $hard(\Phi)$. The cost of a feasible solution $\pi$ for $\Phi$ is determined by the sum of the weights of the soft clauses it falsifies.
\[
cost(\pi, \Phi) = \sum_{C|C \in soft(\Phi) \wedge  \pi \not \models C} wt_C
\]
The goal of solving a {\maxsat} problem $\Phi$ is to minimize the total cost across all soft clauses. An optimal solution can be defined as a feasible solution that satisfies the minimum sum of the weights of the soft clauses. Solving the {\maxsat} problem is also NP-hard, and finding an exact solution for large instances is often impractical. Heuristic algorithms are typically used to find approximate solutions that are close to optimal while balancing the trade-off between satisfying hard constraints and minimizing the penalties for violating soft constraints \cite{stutzle2001review}.

% We will also use $cost(\Phi)$ to denote the cost of an optimal solution
% of $\Phi$, and if $H$ is a set of soft clauses $cost(H)$ to denote the sum of the weights
% of its clauses $cost(H) = \sum_{C \in H} wt_C$.

% Solving the {\maxsat} problem with soft constraints is also NP-hard, and finding an exact solution for large instances is often impractical. Heuristic algorithms are typically used to find approximate solutions that are close to optimal while balancing the trade-off between satisfying hard constraints and minimizing the penalties for violating soft constraints \cite{stutzle2001review}.

% \begin{figure}[h]
%     \centering
%     \includegraphics[page=3, width=0.9\textwidth,trim={0 4.5cm 0 5cm}, clip]{figures/sat_diff.pdf}
%     \caption{Overview of the {\satdiff} Algorithm.}
%     \label{fig:{\satdiff}}
% \end{figure}

\section{Overview of {\satdiff}}
\label{sec:overview}

In this section, we provide an overview of our {\satdiff} framework, a novel approach to tree diffing using the {\maxsat} solving. The {\satdiff} framework takes a pair of source and target code files written in a specific programming language as input and generates a correct and concise edit script. The {\satdiff} framework consists of the following phases:

\textbf{The parsing phase} This phase involves taking a pair of source and target code files written in a specific programming language and parsing them into the source AST $S$ and the target AST $T$. To give a concrete example, let's consider the following pair of codes:
\begin{align*}
&Source:
Sub( Add( A, (Mul( B, C) ),  Sub( D, Sub( E, F) ) ) ) \\
&Target:
Sub( Add( G, Add( B, C) ), Sub( D, A)) 
\end{align*}
which are parsed into a pair of source and target ASTs, as illustrated in Figure \ref{fig:overview_sat1}. Formally, we define an AST as a connected acyclic graph represented by $T:=(N,E)$, where $N,E$ represents a set of nodes and edges respectively. Each node $n \in N$ has a parent node represented by $n.par$, except for the root node, which is denoted as $root(T)$. 
Additionally, any given node $n$ can have an ordered list of child nodes connected through edges. We represent the $i$\textsuperscript{th} child of node $n$ as $n.c_i$, and we denote the connecting edge as $(n,e_i, n.c_i)$. We also denote the $i$\textsuperscript{th} \textit{child slot} or the edge $e_i$ of a (parent) node $n$ as $\halfe{n}{i}$. In addition, each node $n$ possesses both a label ($l:=label(n)$) and an optional string value ($v:=value(n)$). Normally, internal nodes of an AST have a label, whereas a leaf node can have both a label and a value. Note that each node
$n \in N$ is assigned with a unique identifier even if there are nodes that have the same label and value.

We also introduce the concepts of \emph{virtual nodes} and \emph{source space}, both of which play an essential role in constructing an isomorphic tree to the target tree. \emph{Virtual nodes} refer to nodes in the target tree $T$ that do not exist in the source tree $S$, such as nodes 21 and 22 in Figure \ref{fig:overview_sat1}. The \emph{source space} includes the source tree $S$ and all \emph{virtual nodes}. In essence, the target tree $T$ can be reconstructed within the \emph{source space} by removing existing edges and adding new ones.

\begin{figure}
    % \vspace{-5pt}
    \centering
    \subfloat[Parsing phase]{
    \centering
            \includegraphics[page=2, width=0.3\textwidth,trim={0cm 0.5cm 23cm 0cm}, clip]{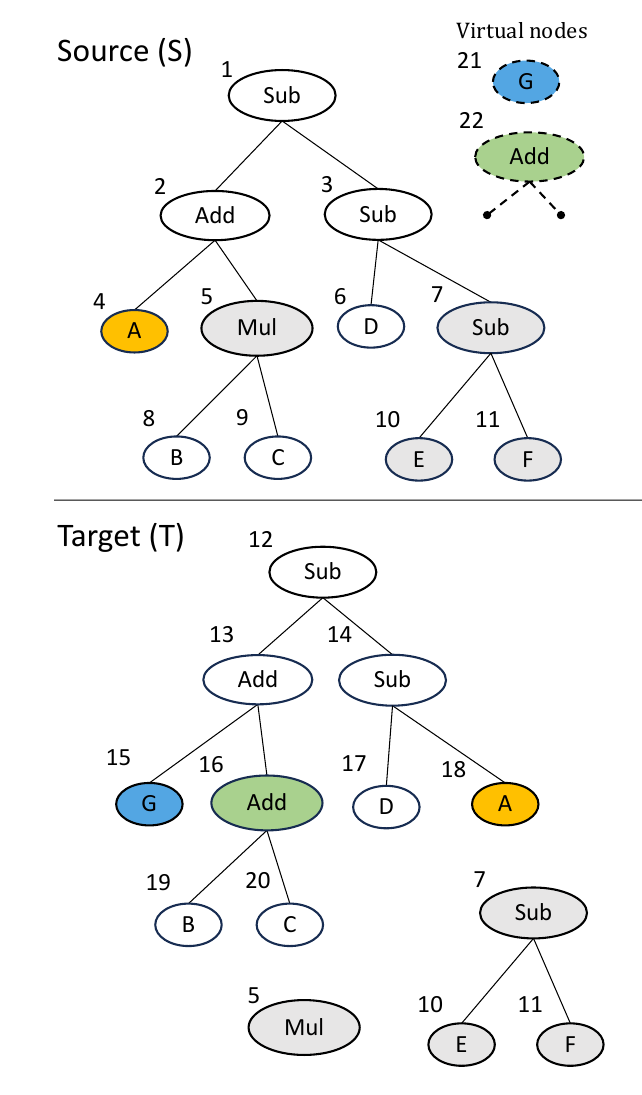}
            \label{fig:overview_sat1}
        }
    \subfloat[Encoding phase]{
    \centering
            \includegraphics[page=2, width=0.31\textwidth,trim={11cm 0.5cm 11.8cm 0.5cm}, clip]{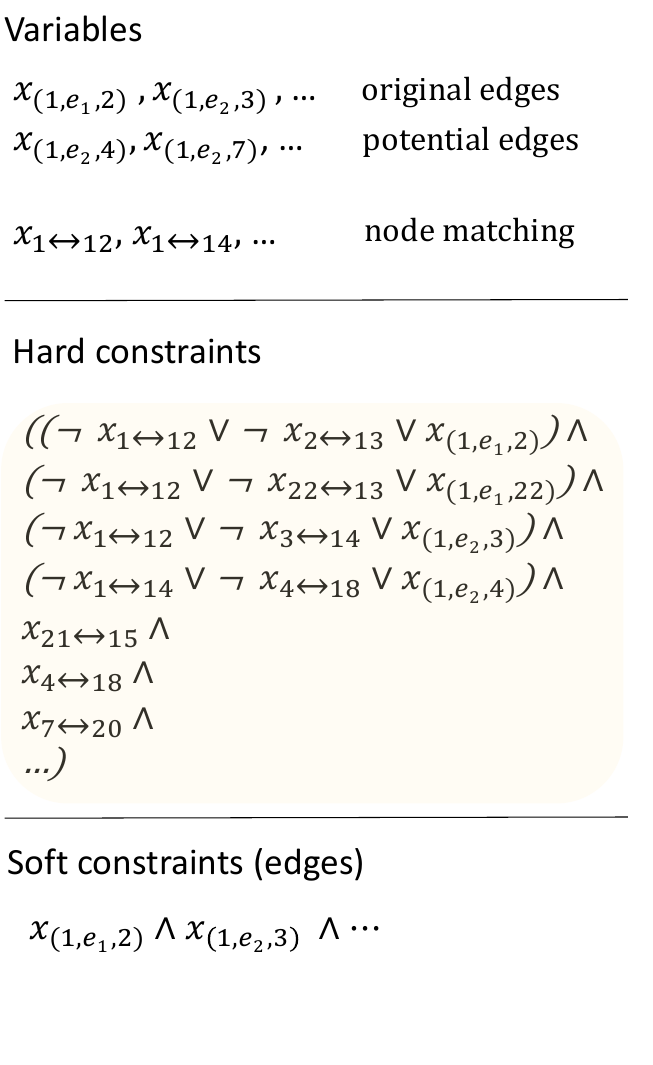}
            \label{fig:overview_sat2}
        }
    \subfloat[Decoding and synthesizing phase]{
    \centering
            \includegraphics[page=2, width=0.33\textwidth, trim={22.2cm 0.5cm 0cm 0.5cm}, clip]{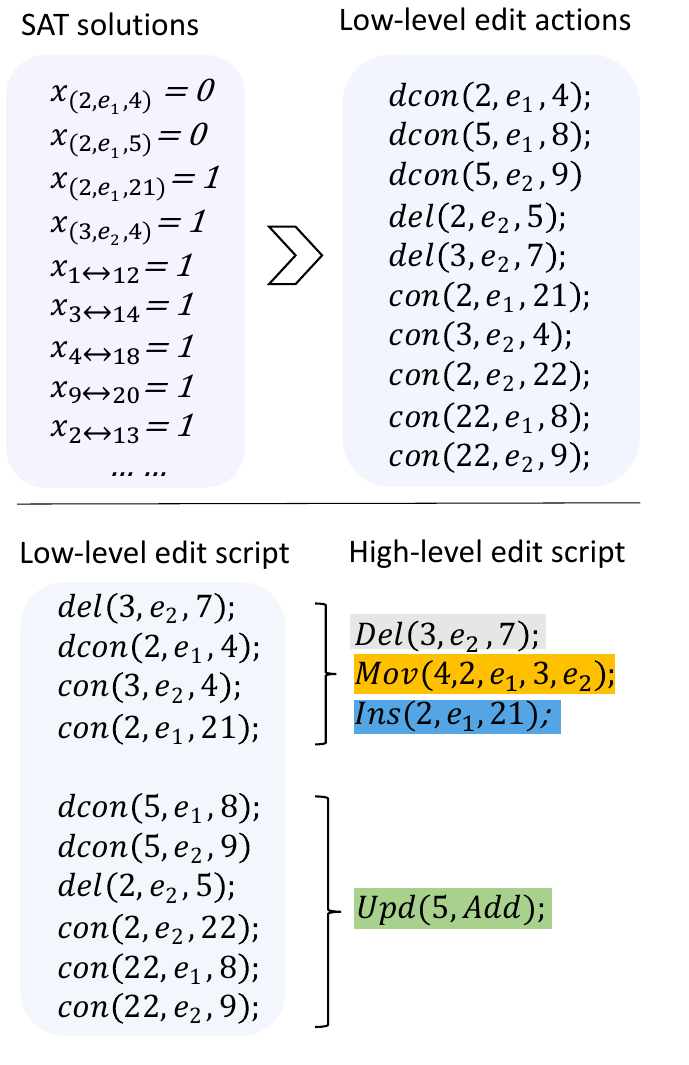}
            \label{fig:overview_sat3}
        }   
    \caption{Overview of the {\satdiff} Algorithm.}
    \label{fig:overview_sat}
    \vspace{-10pt}
\end{figure}

\textbf{The encoding phase} In this phase, we start by generating the necessary propositional Boolean variables, which are then organized into two sets. The first set is designed to encode the state of edges in the \emph{source space}. Formally, we use a variable $\evar{n}{i}{n_c}$ to indicate whether the $i$\textsuperscript{th} edge $e_i$ of node \nd{n} is connected to \nd{n_c}. For instance, $\evar{2}{1}{4}$ encodes the usage of the first edge of the node \nd{2} to \nd{4} in our example. The literal $\evar{n}{i}{n_c}$ indicates that the edge $e_i$ is being utilized, signifying a connection between the corresponding nodes \nd{n} and \nd{n_c}. And the negation $\negevar{n}{i}{n_c}$ indicates that the edge $e_i$ is not being used. It's worth noting that, in addition to existing edges, potential edges (currently nonexistent) are also being encoded.

The second set of variables is used to encode the matching between nodes in $T$ and $S$. Specifically, a variable $\mvar{n}{n'}$ encodes the usage of a match between a node $n$ from the source tree $S$ and a node $n'$ from the target tree $T$. For instance, $\mvar{1}{12}$ encodes the usage of a match between nodes $1$ and $12$. The literal $\mvar{n}{n'}$ indicates that the match between $n$ and $n'$ is established, while the negation $\negmvar{n}{n'}$ indicates that there is no match between $n$ and $n'$. 

% The subsequent step involves encoding a set of hard and soft constraints to address the tree-diffing problem. To provide a conceptual understanding, let's envision a source space comprising the source tree $T$ along with \emph{virtual nodes} --- additional nodes from the target tree $T'$ that do not exist in $T$ (without any edges), as shown in Figure \ref{fig:iso_exmaple1}. By manipulating the connections between nodes in the source space, such as disconnecting or connecting specific edges, it becomes possible to find a tree within the source space that is isomorphic to the target tree, as illustrated in Figure \ref{fig:iso_exmaple}. Isomorphism implies that the two trees are identical in structure except for differences in node identifiers.

% \vspace{-1cm}
The next step involves encoding a set of hard and soft constraints to tackle the tree diffing problem. The hard constraints are designed to guarantee that the tree reconstructed by removing existing edges and adding new edges within the source space is isomorphic to the target tree $T$. Conversely, these soft
constraints specify that both disconnecting existing edges and adding new edges should incur a certain cost.
Intuitively, these soft constraints incentivize the solver to make minimal changes to the structure of the source tree, ultimately contributing to the optimality of the derived edit script. Further details regarding the design of these propositional Boolean variables and constraints will be discussed in Section \ref{sec:phase1}.

\begin{figure}[t]
    % \vspace{-pt}
    \centering
    % \hspace{-2cm}
    \subfloat[Del(3, $e_2$, 7)]{
    \centering
            \includegraphics[page=3, width=0.227\textwidth,trim={0cm 2cm 25.9cm 0cm}, clip]{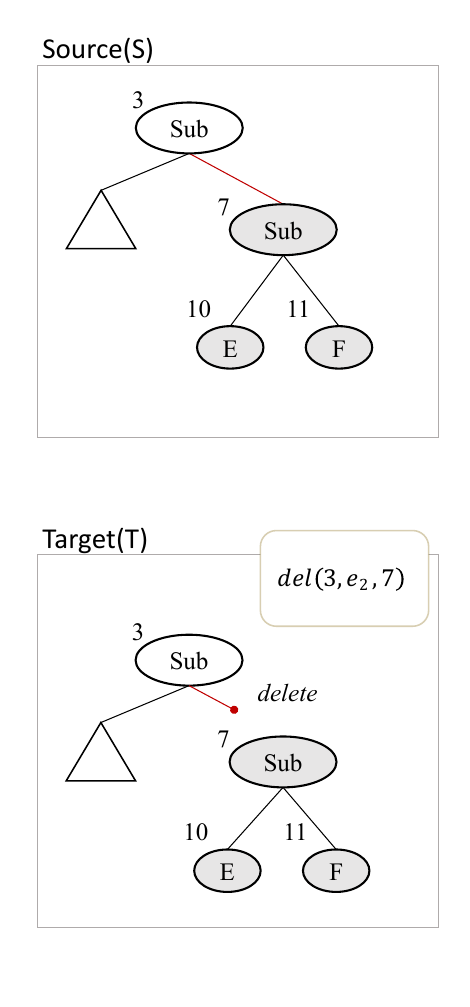}
            \label{fig:high_level_ops1}
        }
    \subfloat[Mov(4, 2, $e_1$, 3, $e_2$)]{
            \centering
            \includegraphics[page=3, width=0.215\textwidth,trim={8.3cm 2cm 18cm 0cm}, clip]{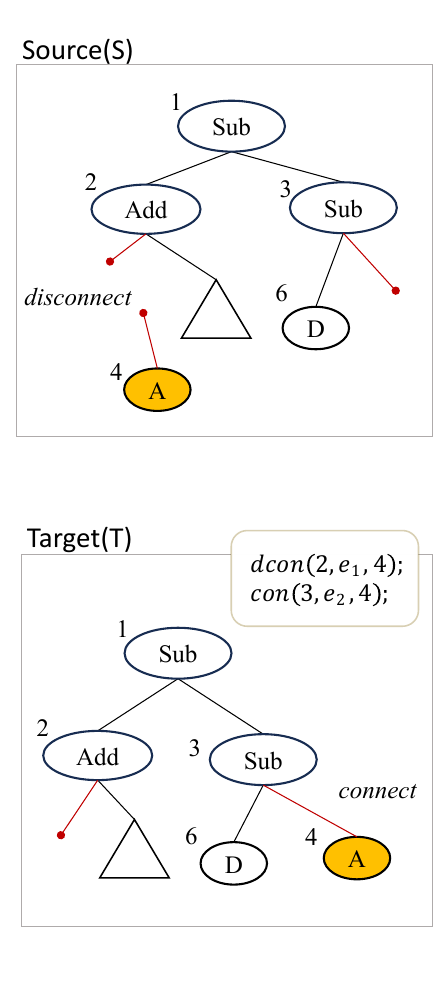}
            \label{fig:high_level_ops2}
        }
    \subfloat[Ins(2, $e_1$, 21)]{
            \includegraphics[page=3, width=0.232\textwidth, trim={16.2cm 2cm 9.5cm 0cm}, clip]{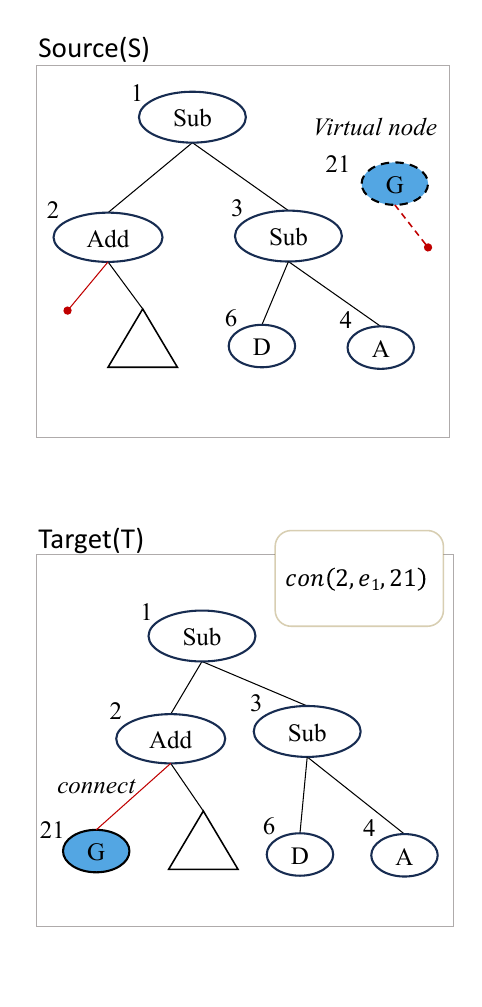}
            \label{fig:high_level_ops3}
        }   
    \subfloat[Upd(5,Add)]{
        \includegraphics[page=3, width=0.265\textwidth, trim={24.5cm 2cm 0cm 0cm}, clip]{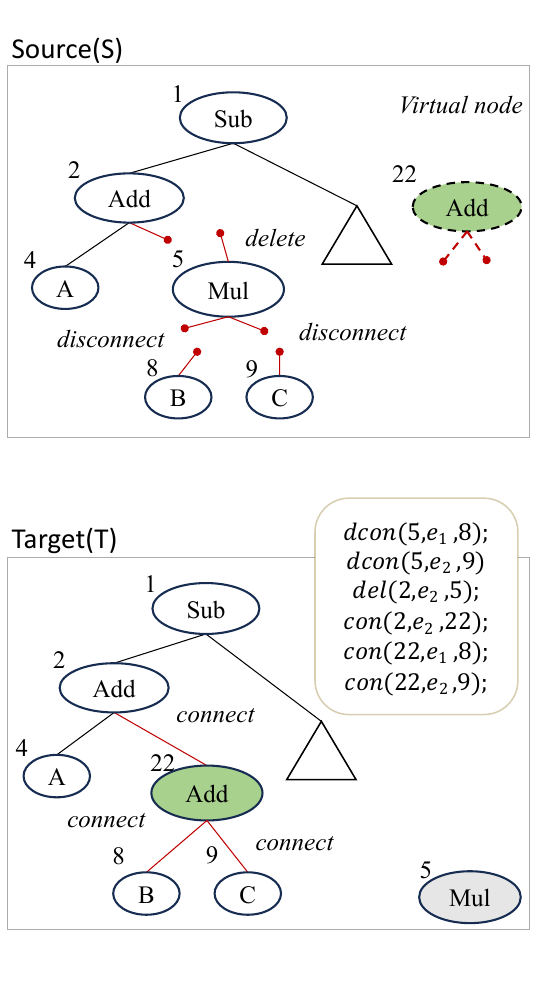}
        \label{fig:high_level_ops4}
        }   
    \caption{High-level Operations.}
    \label{fig:high_level_ops}
    \vspace{-15pt}
\end{figure}

\textbf{The decoding phase} In this phase, our objective is to decode low-level edit actions from SAT solutions. We design three low-level edit actions: \texttt{dcon}, \texttt{del}, and \texttt{con}, representing disconnecting an edge, deleting a node, and connecting an edge, respectively. These actions can be easily inferred from the truth assignment (SAT solution) of our propositional Boolean variables, which encode the state of edges and the usage of node matches.

We demonstrate that low-level edit actions can form a low-level edit script when arranged appropriately. $\satdiff$ provides correctness and minimality guarantees on low-level edit scripts by enforcing hard and soft constraints, respectively. Further details will be discussed in Section \ref{sec:phase2}.

\textbf{The synthesizing phase} In this phase, we synthesize high-level edit scripts from low-level edit actions. It's worth mentioning that high-level edit actions can be composed using specific combinations of low-level edit actions, given an appropriate dependency order, as illustrated in Figure \ref{fig:high_level_ops}. For instance, the process of disconnecting the edge $(2, e_1, 4)$ and then reconnecting node $4$ to node $3$ via a new edge $(3, e_2, 4)$ can be synthesized into a \textit{Move} action. Afterward, connecting the \emph{vrtiual node} $21$ to the slot $(2, e_1)$ of node $2$ can be viewed as an \textit{Insert} action. Additionally, the process of deleting node $5$ while disconnecting nodes $8$ and $9$ from it, which creates an empty slot at the second child of node 2, followed by reconnecting nodes $8$ and $9$ to the \emph{vrtiual node} $22$ and then connecting this subtree to that slot can be merged into a \textit{Update} action.

Specifying an appropriate ordering for low-level edits, as well as synthesizing high-level edit scripts, poses challenges due to complex dependencies. To address this, we leverage dependency graphs that characterize topological orderings of low-level edit actions to generate both low-level and high-level edit scripts. Further details about our approach and other aspects of our high-level edit scripts will be discussed in Section \ref{sec:phase3}.

\begin{algorithm2e}[H]
    \caption{The {\satdiff} framework}
    \label{alg:sat-diff}
    \small
    \DontPrintSemicolon
    % \SetKwInOut{KwIn}{myinput }
    \SetKwProg{Fn}{Function}{}{end}
    \KwIn{source tree\ ($S$),\ target tree\ ($T$)}
    \KwOut{a high-level edit script between the input trees, $\overline{\Delta}$}
    \SetKwFunction{sat}{SATDiff}
    \Fn{\sat{$S$, $T$}}{
    % $virtual\_nodes \gets \text{GetVirtualNodes}(s, t)$ \;
    $v \gets CollectVirtualNodes(S,T)$\tcc*{can be improved with optimization1}
    $M \gets MatchingVariables(S\cup v,T)$\;
    $E \gets EdgeVariables(S\cup v,T)$ \tcc*{can be improved with optimization2}
    $\Phi_{hard} \gets HardEncoding(M,E)$\;
    $\Phi_{soft} \gets SoftEncoding(M,E)$\;
    $\pi$ $\gets$ $\texttt{MAX\_SAT\_Solver}(\Phi_{hard},\Phi_{soft})$\tcc*{SAT solution}
    $\mathcal{E} \gets Decoding(\pi)$ \tcc*{edit effects}
    $dpg \gets ConstructDependencyGraph(\mathcal{E})$\;
    $\overline{\Delta} \gets SynthesizingHighLevelEditScript(dpg)$\;
    \Return{$\overline{\Delta}$}
    }
\end{algorithm2e}

\section{Methodology}
\label{sec:method}

In this section, we discuss the detailed steps of each phase.  Algorithm \ref{alg:sat-diff} outlines the high-level procedure of the {\satdiff} framework. {\satdiff} starts by generating variables to encode the tree diffing problem as a {\maxsat} problem. Subsequently, it utilizes a {\maxsat} problem solver to identify the optimal solution. The optimal solution is then interpreted as tree edits, and, in the final step, high-level edit scripts are synthesized based on the topological ordering of tree edits. 
% To summarize, {\satdiff} initiates by generating variables to encode the tree diffing problem as a {\maxsat} problem. Subsequently, it utilizes a {\maxsat} problem solver to identify the optimal solution. The optimal solution is then interpreted as tree edits, and, in the final step, high-level edit scripts are synthesized based on the topological ordering of tree edits. 

\subsection{The encoding phase --- Generating variables and constraints to encode tree diffing }
\label{sec:phase1}
% In this phase, we show that the tree diffing problem can be encoded into a {\maxsat} problem through a series of steps. 
% The encoding phase is at the core of {\satdiff}, allowing us to specify solution requirements through hard and soft constraints.

% The encoding phase serves as the fundamental component of {\satdiff}, offering the flexibility to define specific solution requirements through the design and imposition of both hard and soft constraints.

% The encoding phase is at the core of {\satdiff}, allowing us to specify solution requirements through hard and soft constraints.

% In this section, we demonstrate how the tree diffing problem is encoded into a {\maxsat} problem, a fundamental step at the core of {\satdiff}, enabling us to define solution requirements through constraints.

\subsubsection{Allocating variables to encode the usage of edges and matches}\leavevmode\newline
To encode the tree diffing problem, we begin by generating a set of propositional Boolean variables that represent modifications to the source tree. We also introduce key definitions, including \emph{virtual nodes}, \emph{source space}, and \emph{edge variables}.

% \begin{definition}[\emph{Virtual nodes}] Let $T$ and $T'$ be a pair of source and target tree, respectively. Let $M_N$ and $M_{N_T}$ represent the multisets of nodes in $T$ and $T'$, respectively. In other words, nodes with the same label and variable are allowed to appear more than once in $M_N$ and $M_{N_T}$. The multiset difference $M_{N_T} \setminus M_N$ is the multiset of nodes that appear in $T'$ but do not appear in $T$. These nodes are defined as \emph{virtual nodes}. We also define the union of $M_{N_T} \setminus M_N$ and $M_N$, $(M_{N_T} \setminus M_N) \cup M_N$, as the \emph{source space}, denoted as $T_S$.
% \end{definition}

% \begin{definition}[\emph{Virtual nodes and the source space}] Consider a pair of source tree $T:=(N,E)$ and target tree $T':=(N_T,E')$, where $N$ and $N_T$ represent multisets of nodes, and $E$ and $E'$ represent multisets of edges. The multiset set difference $N_T \setminus N$ is the collection of \emph{virtual nodes}. We define the graph $S':=(\{N \cup (N_T \setminus N)\} ,E)$ as the \emph{source space}, i.e.,  constructed by $T$ and \emph{virtual nodes}. For simplicity, we also define $\{N \cup (N_T \setminus N)\}  $ as $N_{S'}$. Consequently, we have $S':=(N_{S'}, E)$.
% \end{definition} 

% \begin{remark}
% The multiset difference $N_T \setminus N$ respects the uniqueness of repeated nodes.  For instance, suppose $N_T = \{3,Add,Add,Add\}$ and $N= \{3,Add,4\}$, then \emph{virtual nodes} are $N_T \setminus N = \{Add,Add\}$. To efficiently
% find virtual nodes, we implement Algorithm \ref{alg:get_virtual_nodes}.

% \end{remark}

\begin{definition}[\emph{Virtual nodes and the source space}] Consider a pair of source tree $S:=(N_S,E_S)$ and target tree $T:=(N_T,E_T)$, where $N_S$ and $N_T$ represent multisets of nodes, and $E_S$ and $E_T$ represent multisets of edges. We define \emph{virtual nodes} as copies of all nodes in $N_T$. Then we define \emph{source space} as the graph $S':=(N_S \cup N_T, E_S)$, i.e.,  constructed by $S$ and \emph{virtual nodes}.
\end{definition}

\begin{remark}[\emph{Optimization on virtual nodes}]
Intuitively, we want to construct the target tree $T$ (or construct a tree that is identical to $T$) in the \emph{source space} $S'$. Hence, $S'$ should include all nodes of $T$, i.e., $N_T$. Then by manipulating the edges between nodes in $S'$, it becomes possible to find a tree within $S'$ that is isomorphic to $T$, as stated in Lemma \ref{lemma:iso_exista}.

However, it is sufficient for \emph{virtual nodes} to contain only those nodes that are in the multiset set difference $N_T \setminus N_S$. For instance, suppose $N_T = \{3,Add,Add,Add\}$ and $N_S= \{3,Add,4\}$; then, \emph{virtual nodes} are $N_T \setminus N_S = \{Add,Add\}$. It is trivial to show that adding ${\text{Add}, \text{Add}}$ to $N_S$ provides enough elements to construct the target tree $T$. 

Reducing the number of \emph{virtual nodes} significantly decreases the number of clauses in our encoding, playing an essential role in improving the runtime performance of {\satdiff}, as demonstrated in Section \ref{sec:opt}. 

% To efficiently find the necessary \emph{virtual nodes}, we implement Algorithm \ref{alg:get_virtual_nodes}.

\end{remark}

% \paragraph{Optimization on virtual nodes}
% the multiset set difference $N_T \setminus N_S$. 
% Then we define \emph{source space} as the graph $S':=(\{N_S \cup (N_T \setminus N_S)\} ,E)$, i.e.,  constructed by $S$ and \emph{virtual nodes}. For simplicity, we denote $\{N_S \cup (N_T \setminus N_S)\}  $ as $N_{S'}$. Consequently, we have $S':=(N_{S'}, E)$.
% Then we define \emph{source space} as the graph $S':=(\{N_S \cup (N_T \setminus N_S)\} ,E)$, i.e.,  constructed by $S$ and \emph{virtual nodes}. For simplicity, we denote $\{N_S \cup (N_T \setminus N_S)\}  $ as $N_{S'}$. Consequently, we have $S':=(N_{S'}, E)$.

% \begin{remark}
% The multiset difference $N_T \setminus N_S$ respects the uniqueness of repeated nodes.  For instance, suppose $N_T = \{3,Add,Add,Add\}$ and $N_S= \{3,Add,4\}$, then \emph{virtual nodes} are $N_T \setminus N_S = \{Add,Add\}$. To efficiently
% find virtual nodes, we implement Algorithm \ref{alg:get_virtual_nodes}. \jl{maybe briefly describe what algo1 is about}
% \end{remark}

\begin{definition}[\emph{Edit effect}] 
% An edit effect $\mathcal{E}(S,T)$ is a four-element tuple 
% $(\mathcal{E}_{N+}(N_S,N_T), \mathcal{E}_{N-}(N_S,N_T), \mathcal{E}_{E+}(E_S,E_T), \mathcal{E}_{E-})(E_S,E_T)$ 

% representing necessary edits for source tree $S$ to become $T$, i.e., $\mathcal{E}(S,T) \circ S = T$. 

An edit effect $\mathcal{E}(S,T)$ is a four-element tuple $(\mathcal{E}_{N+}(S,T),  \mathcal{E}_{N-}(S,T), \\ \mathcal{E}_{E+}(S,T), \mathcal{E}_{E-}(S,T))$ representing the necessary edits for the source tree $S$ to transform into the target tree $T$. To be more specific, $\mathcal{E}_{N+}, \mathcal{E}_{N-}, \mathcal{E}_{E+}, \mathcal{E}_{E-}$ denote the nodes and edges that need to be added and deleted, respectively. We intentionally omit $S$ and $T$ when the context is clear.

% i.e., $\mathcal{E}(S,T) S = T$. 

% \begin{align*}
%     \mathcal{E}(S,T) &:=  (\mathcal{E}_{N+}(S,T), \mathcal{E}_{N-}(S,T), \mathcal{E}_{E+}(S,T), \mathcal{E}_{E-}(S,T)) \\
%     \textit{where  } \text{  } \mathcal{E}_{N+}(S,T)  &=  N_T \setminus N_S \text{  }\text{  }\text{  }\text{  }\text{  }
%                                \mathcal{E}_{N-}(S,T)  = N_S \setminus (N_T \cap N_S) \\
%                                \mathcal{E}_{E+}(S,T) &=  E_T \setminus E_S \text{  }\text{  }\text{  }\text{  }\text{  }\text{  }
%                                \mathcal{E}_{E-}(S,T)  =  E_S \setminus (E_T \cap E_S)
% \end{align*}

% composed by an \emph{edge edit effect} $\mathcal{E}_E$ and a \emph{node edit effect} $\mathcal{E}_N$, i.e., $\mathcal{E} := (\mathcal{E}_N, \mathcal{E}_E)$.
% For a tree $T:=(N, E)$, $\mathcal{E}_E$ is a mapping from an edge set $E_0$ to an arbitrary edge configuration $E'$, where $E_0,E' \in  \mathcal{P}(N \times N)$. Formally, we define :
% \[
% \mathcal{E}_E(E_0,E') = \begin{cases}
%     E_0 \cup \{e\} &  \text{if } e \in E' \setminus E_0 \\
    
%     E_0 \setminus \{e\} & \text{if } e \in (E_0 \setminus (E_0 \cap E'))
% \end{cases}
% \]
% Whereas $\mathcal{E}_N$ is a mapping from a vertex set $N_0$ to any possible vertex set $N_T$, where $N_0, N_T \subset N$. Formally, then we define :
% \[
% \mathcal{E}_N(N_0,N_T) = \begin{cases}
%     N_0 \cup \{n\} &  \text{if } n \in N_T \setminus N_0 \\
%     N_0 \setminus \{n\} & \text{if } n \in (N_0 \setminus (N_0 \cap N_T))
% \end{cases}
% \]
% In summary, $\mathcal{E}(T,T')  =(\mathcal{E}(N_0, N_T), \mathcal{E}(E_0, E')) = (N_T, E')$. \\
\end{definition} 
\begin{remark}

Note that the edit effect $\mathcal{E}$ solely describes the outcome of the edit, without specifying a sequence of edit actions. Still, it serves as a guideline to generate our edit scripts. Intuitively, we view an edit script as a sequence of actual edit actions designed to achieve the editing effect described by the corresponding edit effect, as we will discuss later. In addition, the creation of \emph{source space} $S'$ corresponds to $\mathcal{E}_{N+}(S,T)$. Then we only need an edit effect $\mathcal{E}(S',T)$ that consists of three-element tuple $(\mathcal{E}_{N-}(S',T),\mathcal{E}_{E+}(S',T),\mathcal{E}_{E-}(S',T))$ to describe necessary edits that applied to the source space $S'$ to transform into $T$.

% Meanwhile,
% we believe nodes are 
% Intuitively, $\mathcal{E}_E$ adds edges that appear in $E'$ but not in $E$, and it deletes edges that appear in $E$ but not in $E'$. The set of common edges between $E$ and $E'$ remains unchanged by $\mathcal{E}_E$. The same argument also applies to $\mathcal{E}_N$.
\end{remark}

\begin{lemma}
\label{lemma:iso_exista}
There exists at least one edit effect $\mathcal{E}$ that applies to the source space $S':=(N_{S'}, E_S)$, resulting in a tree $T_\mathcal{E}$ such that $T_\mathcal{E}$ is isomorphic to the target tree $T$. 
% There exists at least an edit effect $\nabla$ that applies to source space $S':=(N_{S'} ,E)$ such that the resulting tree $T_\nabla:= \nabla \circ S'$ so that $T_\nabla$ is isomorphic to the target tree $T'$. Formally, we have:
% \begin{align*} 
% \exists \mathcal{E}: T_\mathcal{E}:=\mathcal{E} \circ S' \quad \text{s.t.} 
% \quad T_\mathcal{E} \cong T 
% \end{align*} 
% where $T_\mathcal{E} = (N_\mathcal{E}, E_\mathcal{E})\text{, } N_\mathcal{E} = \mathcal{E}_{N-}(N_{S'}) = N_T,  E_\mathcal{E} = \mathcal{E}_N(E) = E'$. 
\end{lemma}
\begin{proof}
A trivial edit effect $\mathcal{E}(S',T)$ can be described as follows: $\mathcal{E}_{E-}(S',T) = E_S$, $\mathcal{E}_{E+}(S',T) = E_T$, and $\mathcal{E}_{N-}(S',T) = N_{S'} \setminus N_T$. First, we delete edges in $\mathcal{E}_{E-}(S',T)$ in $S'$, resulting in a modified graph $S' := (N_{S'}, \varnothing)$. Next, we 
construct edges in $\mathcal{E}_{E+}(S',T)$, resulting in $S' := (N_{S'}, E_T)$. Finally, the remaining unconnected nodes that lies in $\mathcal{E}_{N-}(S',T)$ are deleted, resulting in $S' := (N_T, E_T)$.
\end{proof}

To encode the edit effect from the source space $S'$ to the target tree $T$, we require \emph{match variables} and \emph{edge variables}. Recall that our edge notions: We denote the $i$\textsuperscript{th} child slot or the edge $e_i$ of a (parent) node $n$ as $(n, e_i)$. An edge can also be denoted as $(n, e_i, n_c)$, interpreted as the $i$\textsuperscript{th} edge of $n$ connecting to $n_c$. Then we define \emph{edge variables}.

% \vspace{-1cm}

% \input{algo1}

\begin{definition}[\emph{The edge variable}] Let $n$ and $n_c$ be two nodes in the source space $S'$, a variable $\evar{n}{i}{n_c}$ encodes the usage of the edge $e_i$ in the connection between  node $n_c$ and (parent) node $n$. 
\end{definition}

\begin{remark}
The literal $\evar{n}{i}{n_c}$ indicates that the edge $e_i$ is being utilized, signifying a connection between the corresponding nodes $n$ and $n_c$ through edge $e_i$. In other words, node $n_c$ uses the edge $(n, e_i)$, becomes the $i$\textsuperscript{th} child of node $n$. On the other hand, the negation $\negevar{n}{i}{n_c}$ indicates that the edge $e_i$ is not being used, although nodes $n$ and $n_c$ may be disconnected or connected through other edges. 
\end{remark}

\begin{definition}[\emph{The match variable}] 
The set of nodes matched with $n$ in $T$ is defined as $M_{node}(n,N_T) := \{ n' \mid n' \leftrightarrow n, n' \in N_T\}$, where $n \leftrightarrow n'$ indicates that the following match conditions are met:
\begin{align*}
    &label(n) = label(n'), value(n) = value(n')
\end{align*}
We also define the set of match variables with $n$ in $T$ as $M_{var}(n,N_T) := \{\mvar{n'}{n} \mid n' \in M_{node}(n,N_T) \}$. Conversely,
the set of nodes matched with  $n'$ in $S'$ is denoted as  $M_{node}(n',N_{S'}) := \{ n \mid n \leftrightarrow n', n \in N_{S'}\}$, and and the set of match variables with $n'$ in $T$ as $M_{var}(n',N_{S'}) := \{\mvar{n}{n'} \mid n' \in M_{node}(n',N_{S'}) \}$.

\end{definition}

\begin{remark}
The literal $\mvar{n}{n'}$ 
represents a match between
$n$ and $n'$, while the negation $\negmvar{n}{n'}$ signifies no match. 
Multiple match variables can be created for a node $n$ when it has several potential matching nodes in $T$, but only one will be used and determined by the \maxsat{} solver.
\end{remark}

% \begin{remark}
% The match condition between two edges $(n, e_i, n_c)$ and $(n', e_i, n_c')$
% is denoted as: $(n, e_i, n_c) \leftrightarrow (n', e_i, n_c') \text{ if and only if } n \leftrightarrow n', n_c \leftrightarrow n_c'$. 
% \end{remark}
% The first step is to identify the nodes that the source tree lacks to be transformed into the target tree. As shown in Algorithm\ref{alg:get_virtual_nodes}, we utilize two Hashmaps as histograms for the two given trees, with default key values set to zero. By comparing these histograms, we collect nodes that are necessary for the transformation, referred to as the virtual nodes, whenever we observe a higher node count in the target tree's histogram. By combining these virtual nodes with the source tree, denoted as the source space, the SAT solver now have a working space with all the necessary components to determine the optimal transformation between the target and source trees on an atomic edge level.

% Recall that our edge notions: We denote the $i$\textsuperscript{th} child slot or the edge $e_i$ of a (parent) node $n$ as $(n, e_i)$. An edge can also be denoted as $(n, e_i, n_c)$, interpreted as an $i$\textsuperscript{th} edge of $n$ connecting $n_c$. Then we define \emph{potential edges}.

% as all possible edges excluding existing edges, i.e., $(N_{S'} \times  N_{S'} )\E_S$, denoted as $\widetilde{E}((n',e_i))$.

\begin{definition}[\emph{Potential edges}]
Let $n'$ be a node from the target tree $T$, and 
$(n',e_i) \in E_T$ be the $i^{th}$ edge of $n'$. Then we define the potential edges matched with $(n',e_i)$ as all possible edges excluding existing edges in the source space in $S'$, i.e., $(N_{S'} \times N_{S'}) \setminus E_S$, denoted as $\widetilde{E}((n',e_i))$. The set of potential edges matched with every edge in $E_T$ is called the potential edges (set), denoted as ${\bigcup_{(n',e_i) \in E_T}  \widetilde{E}((n',e_i))}$.
\end{definition}

% follows:
% \begin{align*}
%  \widetilde{E}((n',e_i)): = \{  &(n,e_i,n_c) \mid n \leftrightarrow n', n_c \leftrightarrow n'.c_i, \forall \text{} n,n_c \in N_{S'} \}
% \end{align*}

% \begin{remark}
% The match condition between two edges $(n, e_i, n_c)$ and $(n', e_i, n_c')$
% is denoted as: $(n, e_i, n_c) \leftrightarrow (n', e_i, n_c') \text{ if and only if } n \leftrightarrow n', n_c \leftrightarrow n_c'$. 
% \end{remark}

\begin{remark}
\label{rk:opt_on_edges}[\emph{Optimization on potential edges}]
\emph{Potential edges} refer to the collection of edges that are not present in the source space $S'$ but may be needed for constructing an isomorphic tree to $T$ in $S'$. It is clear that we only need a small subset of $(N_{S'} \times N_{S'}) \setminus E_S$ to search for potential edges matched with $(n',e_i)$. To prevent an explosion in the number of all possible edges, we filter them by leveraging the \emph{label} and \emph{value} information of the target tree. This information, extracted from the edges of the target tree, provides strict rules for selecting child nodes for a specific parent node. Formally, we could define the filtered \emph{potential edges} matched with $(n',e_i)$ as follows:
\begin{align*}
 \widetilde{E}((n',e_i)): = \{  &(n,e_i,n_c) \mid n \leftrightarrow n', n_c \leftrightarrow n'.c_i, \forall \text{} n,n_c \in N_{S'} \}
\end{align*}

Such a reduction in the number of \emph{potential edges} significantly mitigates the clause explosion in the encoding phase, thus improving the runtime performance of {\satdiff} dramatically, as demonstrated in Section \ref{sec:opt}. 

% To efficiently find filtered \emph{potential edges}, we implement Algorithm \ref{alg:get_pot_edges}.
\end{remark}

\subsubsection{Encoding  the tree diffing problem using hard and soft constraints} \leavevmode\newline
We show the tree diffing problem can be encoded into a max satisfiability problem represented by a CNF $\Phi_{{\satdiff}}$. To facilitate our discussion on $\Phi_{{\satdiff}}$, we introduce the following cardinality constraints \cite{Biere2009}:

% Definition \ref{lemma:at_most_one} and \ref{lemma:exactly_one}.
% The tree diffing problem in our settings can be encoded with a set of hard and soft constraints. Formally,

\begin{definition}
\label{lemma:at_most_one}
[\emph{The at most one (AMO) constraint}]: Given a collection of variables $ X = \{x_1, \ldots, x_n\}$, for $1 \leq i \leq n$, the SAT encoding of the AMO constraint
is the following:

\[
\Phi_{\leq1}(X) = \neg s_0 \bigwedge \left(\bigwedge_{i=1}^{n} \left(\left(\neg s_{i-1} \lor s_{i}\right) \land \left(\neg s_{i-1} \lor \neg x_i\right) \land \left(s_{i-1} \lor \neg x_i \lor s_{i}\right)\right)\right) 
\]

% where $s_i = \bigwedge_{j=1}^{i} x_j$, and $s_0 = false$.
% where $s_i = \displaystyle\bigwedge\limits_{j=1}^{i} x_j$.

% \xs{$s_i$ should be $s_{i-1}$, $s_{i+1}$ should be $s_{i}$; remove either $s_0 = false$ or $ \neg s_0 $}
% This conjunctive normal form contains $\mathcal{O}{(2n+1)}$ literals and $\mathcal{O}(3n)$ clauses.

\end{definition}

\begin{definition}
\label{lemma:exactly_one}
[\emph{The exactly one (EO) constraint}]: 
Given a collection of variables $ X = \{x_1, \ldots, x_n\}$, for $1 \leq i \leq n$, the SAT encoding of the EO constraint is the following:
\[
\Phi_{=1}(X) = \left(\bigvee_{i=1}^{n} l_i\right) \land \Phi_{\leq1}(X)
\]
% This conjunctive normal form contains $\mathcal{O}{(2n+1)}$ literals and $\mathcal{O}(3n+1)$ clauses.

\end{definition}

Note the variables $X$ for $\Phi_{{\satdiff}}(X)$ consist of the edge variables and match variables introduced earlier, denoted as $X:=\{\mvar{n}{n'}|n \in N_{S'}, n' \in N_T \} \cup \{\evar{n}{i}{n_c}|(n,e_i,n_c) \in \{\bigcup\limits_{(n',e_i) \in E_T}  \widetilde{E} ((n',e_i)) \bigcup E_S\}\}$.
\begin{align*}
    \Phi_{{\satdiff}} &= hard(\Phi_{{\satdiff}}) \bigwedge soft(\Phi_{{\satdiff}}) \\
    hard(\Phi_{{\satdiff}}) &= \Phi_{SN} \land \Phi_{TN}\land \Phi_{CN} \land \Phi_{PN} \land \Phi_{RN} \land \Phi_{\cong} \\
    soft(\Phi_{{\satdiff}}) &= \Phi_{DE} \land \Phi_{AE}
\end{align*} 

where
% \begin{align*}
%     &\Phi_{SN}: \emph{\text{Each node in the source space has at most one target node match.}}  \\
%     &\[  \Phi_{SN} = \bigwedge_{n \in \{N \cup (N_T \setminus N)\} } \Phi_{AMO}(M_{var}(n,T'))\]
%     &\Phi_{TN} : \emph{\text{Each target node has exactly one match node in the source space.}}  \\
%     &\Phi_{CN} : \emph{\text{Each node in source space has at most one parent node.}}  \\
%     &\Phi_{PN} : \emph{\text{Each parent node in source space have at most one $i$\textsuperscript{th} child node.}}  \\
%     &\Phi_{\cong} : \emph{\text{When both the parent and child have matches in the target tree, their}} \\ &\emph{\text{matched nodes in the source space must be connected by the corresponding edge.}} \\
%     &\Phi_{RN} :  \emph{\text{If any node matches with the root of the target tree, then the node has no parent. }}  \\
%     &\Phi_{DE} : \emph{\text{Encourage using existing edges.}}  \\
%     &\Phi_{AE} : \emph{\text{Discourage the utilization of potential edges.}} \\
% \end{align*} 

\begin{enumerate}
  \item $\Phi_{SN}$: Each node $n$ in the source space has at most one target node match. \\
  \[  \Phi_{SN} = \bigwedge_{n \in N_{S'} } \Phi_{\leq1}(M_{var}(n,N_T))\]

where $N_{S'}$ is the set of nodes of the source space $S'$, $M_{var}(n,N_T)$ represents the set of  variables in $N_T$ that match the source node $n$.\\
% \textit{\underline {Clause complexity}:} A source node can have no more than $|N_T|$ matches. By applying Lemma \ref{lemma:at_most_one} to every source node, we obtain the upper bound for this constraint: $\mathcal{O}(3|N_T|^2+3|N||N_T|)$. \\
  % Encode node in source space. For each node in source space we collect a set of all possible matching in target space mentioned in the pre-processing step, denoted as $[m_1, m_2, \dots, m_i]$. Each of the matching sets will then be input to our helper function \texttt{at\_most\_one} to ensure each node has at most one matching in target space. Then a CNF will be returned for the MAXSAT solver to solve.
  \item $\Phi_{TN}$: Each target node $n'$ has exactly one match node in the source space.  \\
  \[  \Phi_{TN} = \bigwedge_{n' \in N_T}  \Phi_{=1}(M_{var}(n',N_{S'}))\] 

where $M_{var}(n',N_{S'})$ represents the set of variables in  matched with target node $n'$. \\
% \textit{\underline {Clause complexity}:} 
% A target node can have no more than $|N|+|N_T|$ matches in $S'$. By applying Lemma \ref{lemma:exactly_one} to all $|N_T|$ target nodes, the upper bound for this constraint is $\mathcal{O}(3|N_T|^2+3|N||N_T|+|N_T|)$. \\

% A source node can have no more than $|N_T|$ matches. Applying Lemma\ref{lemma:at_most_one} to every source node, we get the upper bound of this constraint = $\mathcal{O}(3|N_T|^2+3|N||N_T|)$.
  % Encode node in target space. For each node in target tree we collect matching literals as in step 1. However, this time the matching sets will be input to another helper function called \texttt{one\_and\_only\_one}, which guarantees one and only one matching in the source space.
  
  \item \textit{$\Phi_{CN}$: Each child node $n_c$ in the source space has at most one parent node.} \\
  \[  \Phi_{CN} = \bigwedge_{n_c \in N_{S'}}  \Phi_{\leq1}(\{\evar{n}{i}{n_c}|(n,e_i,n_c) \in \{\bigcup_{(n',e_i) \in E_T}  \widetilde{E} ((n',e_i)) \bigcup E_S\}\})\]

where ${\bigcup\limits_{(n',e_i) \in E_T}  \widetilde{E} ((n',e_i))}$ is the set of potential edges in the source space $S'$. \\
% \textit{\underline {Clause complexity}:} 
% A node can have no more than $|N|+|N_T|-1$ $i^{th}$ edges connecting to the parent node, where $i$ is less than or equal to some constant $\emph{C}$. By applying Lemma \ref{lemma:at_most_one} to $C(|N|+|N_T|-1)$ edge choices for all $|N|+|N_T|$ source nodes, we can determine the upper bound for the number of clauses in this constraint: $\mathcal{O}(3C(|N|+|N_T|-1)(|N|+|N_T|)) \approx \mathcal{O}((|N|+|N_T|)^2)$. \\

  % Encode child node in source space. For node in source space, we first use python grammar to find all the valid edge literals, denoted as $[e_1, e_2, \dots, e_j]$. And the literal set are input to \texttt{at\_most\_one} to ensure each child node has at most one parent in source space.

  \item $\Phi_{PN}$: Each parent node $n$ in source space have at most one $i$\textsuperscript{th} child node. \\
  \[  \Phi_{PN} = \bigwedge_{n \in N_{S'}}  \Phi_{\leq1} (\evar{n}{i}{n_c}| (n,e_i,n_c) \in \{\bigcup_{(n',e_i) \in E_T}  \widetilde{E} ((n',e_i)) \bigcup E_S\}  \})\]     
% where $\{N \cup (N_T \setminus N)\}$ is the set of nodes of the source space $S'$, ${\bigcup_{(n',e_i) \in E}  \widetilde{E} ((n',e_i))}$ is the set of potential edges of $S'$. \\
% \textit{\underline {Clause complexity}:} 
% By following the same argument as in $\Phi_{CN}$, the upper bound for the number of clauses of this constraint: $\mathcal{O}(3C(|N|+|N_T|-1)(|N|+|N_T|)$ $\approx$ $\mathcal{O}((|N|+|N_T|)^2)$. \\

% A node can have no more than $|N|+|N_T|-1$ number of $i^{th}$ child choices, in which $i$ is less than or equal to some constant $\emph{C}$. Applying Lemma \ref{lemma:at_most_one} to $C(|N|+|N_T|-1)$ child choices for all $|N|+|N_T|$ source nodes, we get the upper bound for the number of clauses of this constraint: $\mathcal{O}(3C(|N|+|N_T|-1)(|N|+|N_T|)$ $\approx$ $\mathcal{O}((|N|+|N_T|)^2)$. \\

  % Encode parent node in source space. In contrast to step 3, here we ensure that one parent has at most one $k^{th}$ child node. Hence, we input edge literals into \texttt{at\_most\_one} for $k$ times where k is the number of children of some parent node $p$.

  \item $\Phi_{\cong}$: [\emph{The isomorphism constraint}] When both the parent node $n'$ and the child node $n_c'$ in the target tree have matches, their matched nodes $n$ and $n_c$ in the source space $S'$ must be connected by the corresponding edge $e_i$.
  % For every edge in source space, written as $E(P, k^{th}, C)$ where $P$ stands for parent node and $C$ for child node. Next, we obtain all matching for parent P using helper function \texttt{find\_matches(P, source)} and all matching for child C using \texttt{find\_matches(C, source)}, denoted as set of Boolean variables $[m_1, m_2, \dots, m_i]$ for parent matching and $[n_1, n_2, \dots, n_j]$ for child. After that for each $(m_a, n_b)$ pair, if $m_a \wedge n_b$ outputs \texttt{true}, that is, when both the parent and child have matching in target space, there must exist $Edge(P_a, k^{th}, C_b)$ between their corresponding matched nodes.

    \[  \Phi_{\cong} = \bigwedge_{ (n',e_i,n'_c) \in E_T}  \{\evar{n}{i}{n_c}|n \in M_{node}(n',N_{S'}), n_c \in M_{node}(n'_c,N_{S'})\}\]

    % \[  \Phi_{\cong} = \bigwedge_{n \in M_{node}(n',N_{S'}), n_c \in M_{node}(n'_c,N_{S'})}
    %  \{\evar{n}{i}{n_c}|(n',e_i,n'_c) \in E_T\}\] 
     
    % For every edge in target space, written as $E(P,k^{th},C)$, get the matching set for P and C respectively, denoted as $m_i$, $n_j$. Then for each $(m_a, n_b)$ matching pair, $Edge(P_a, k^{th}, C_b)$ must be true. 
where $M_{node}(n',N_{S'})$ and $M_{node}(n'_c,N_{S'})$ are sets of nodes in the source space matched with $n'$ and $n'_c$, respectively.\\
% \textit{\underline {Clause complexity}:} 
% The parent and child nodes of each edge in the target tree can have $|N_T|^2$ match variables. Since the target tree has $|N_T|-1$ edges, the upper bound for this constraint is $\mathcal{O}(|N_T|^3-|N_T|^2)$. \\ 

\item $\Phi_{RN}$: If any node matches with the root of the target tree $T$, then the node has no parent. \\
    \[  \Phi_{RN} = \bigwedge_{n_c \in M_{node}(root(T),N_{S'})}  \{\negevar{n}{i}{n_c}|(n,e_i,n_c) \in \{\bigcup\limits_{(n',e_i) \in E_T}  \widetilde{E} ((n',e_i)) \bigcup E_S\}\}\] 

where $M_{node}(root(S'), N_{S'})$ is the set of nodes in the source space $S'$ that match with the root of $T$.\\
% \textit{\underline {Clause complexity}:} By following the same argument as in $\Phi_{CN}$, the upper bound for the number of clauses of this constraint: $\mathcal{O}(3C(|N|+|N_T|-1)(|N|+|N_T|)$ $\approx$ $\mathcal{O}((|N|+|N_T|)^2)$. \\

  % An encoding for an edge case. If any node matches with the root of target tree then the node has no parent. 

    \item $\Phi_{DE}$: Encourage the utilization of exisiting edges.
     \[  \Phi_{DE} = \bigwedge_{(n,e_i,n_c) \in E_S} \evar{n}{i}{n_c} \] 
% \textit{\underline {Clause complexity}:} The maximum number of edges that can appear in the source tree $T$ is $|N|-1$. Therefore, the upper bound for the number of clauses for this constraint is $\mathcal{O}(|N|-1)$.
% \\
    
    % An encoding for an edge case. If any node matches with the root of target tree then the node has no parent. 
    \item $\Phi_{AE}$: Discourage the utilization of potential edges. 
      
     \[  \Phi_{AE} = \bigwedge_{(n,e_i,n_c) \in \bigcup\limits_{(n',e_i) \in E_T}  \widetilde{E}((n',e_i))}  \negevar{n}{i}{n_c}\]   

% \textit{\underline {Clause complexity}:} The maximum number of potential edges that can appear in the source space $S'$ is $(|N|+|N_T|)^2$. Therefore, the upper bound for the number of clauses of this constraint is $\mathcal{O}(|N|^2+|N_T|^2+2|N||N_T|)$.

      % An encoding for an edge case. If any node matches with the root of the target tree then the node has no parent. 
\end{enumerate}

% We show that 
% the solution $\pi$ to the $\Phi_{{\satdiff}}$ can be interpreted as an edit effect $\mathcal{E}^{\pi}$ from the source space $S'$ to some graph $G$.

\paragraph{Interpret the solution to $\Phi_{{\satdiff}}$ as edit effect.} Recall our encoding variables are $X:=\{\mvar{n}{n'}\} \cup \{\evar{n}{i}{n_c}\}$. Then variables assigned by a solution $\pi$ can be interpreted as an edit effect $\mathcal{E}^{\pi}:=(\mathcal{E}_{N-}^{\pi},\mathcal{E}_{E+}^{\pi},\mathcal{E}_{E-}^{\pi})$ that applies to the source space $S':=(N_{S'} ,E_S)$:
\begin{align*}
\mathcal{E}_{N-}^{\pi} &:=  \{n \text{ } | \text{ }\text{ } \pi(\mvar{n}{n_c}) = 0\} \\
\mathcal{E}_{E+}^{\pi} &:= 
\{(n,e_i,n_c)  \text{ } | \text{ }\text{ } (n,e_i,n_c) \in \bigcup_{(n',e_i) \in E_T}  \widetilde{E}((n',e_i)), \pi(\evar{n}{i}{n_c}) = 1\} \\
\mathcal{E}_{E-}^{\pi} &:=
\{(n,e_i,n_c) \text{ }  |  \text{ }\text{ } (n,e_i,n_c) \in E_S,\text{ } \pi(\evar{n}{i}{n_c}) = 0\} 
\end{align*}
Note that the solution $\pi$ may be infeasible, failing to satisfy the hard constraints $hard(\Phi_{{\satdiff}})$. In such cases, applying the edit $\mathcal{E}^{\pi}$ to $S'$ could result in an arbitrary graph $G:=(N_G,E_G)$, where $N_G \subseteq N_{S'}$ and 
$E_G \in \mathcal{P}(N_{S'} \times N_{S'})$. In addition, the solution $\pi$ doesn't directly specify the ordering of adding and removing nodes and edges; it is limited to indicating which edges and nodes should be modified, corresponding to an edit effect $\mathcal{E}^{\pi}$.

\begin{theorem} [\emph{Guarantee of correctness}] 
\label{thm:sol_correct}
Any solution $\pi$ that satisfies the hard constraints, i.e., $\pi \models hard(\Phi_{{\satdiff}})$, implies a correct edit effect $\mathcal{E}^{\pi}$ from the source space $S'$ to the target tree $T$. The correctness of edit effect $\mathcal{E}^{\pi}$ ensures the existence of a tree $T_\mathcal{E}$ in the source space that is isomorphic to the target tree $T$, i.e., $T_\mathcal{E} \cong T$. 
\end{theorem}

\begin{proof}
We present a sketch of the proof.

% To demonstrate $T_\Delta \cong T'$, it suffices to show the existence of a matching between the vertices of $T_\Delta$ and $T'$, where two vertices are connected by an edge in $T_\Delta$ if and only if their corresponding vertices are connected by an edge in $T'$. This matching condition is enforced by our $\Phi_{EM}$. 

The constraints $\Phi_{SN}$, $\Phi_{RN}$, and $\Phi_{TN}$ ensure the validity of the established matches between the nodes of the source space $S'$ and the target tree $T$. Similarly, the constraints $\Phi_{CN}$ and $\Phi_{PN}$ guarantee the validity of tree candidates from the modified source space $S'$. The isomorphism constraint $\Phi_{\cong}$ enforces a matching between the nodes of any tree candidate and $T$. In this matching, an edge connects two nodes in the candidate if and only if their corresponding nodes are connected by an edge in $T$. Ultimately, any resulting tree candidate satisfying all constraints in $hard(\Phi_{{\satdiff}})$ is a valid tree $T_\mathcal{E}$ that satisfies $T_\mathcal{E} \cong T$. To conclude, any solution $\pi$ s.t. $\pi \models hard(\Phi_{{\satdiff}})$ infers a correct $\mathcal{E}^{\pi}$.
\end{proof}

\begin{theorem} [\emph{Guarantee of minimality}]
\label{thm:sol_opt}
By optimizing the soft constraints  $soft(\Phi_{{\satdiff}})$,
\[
\min_{\pi} cost(\pi,\Phi_{{\satdiff}})
\]
we obtain the optimal solution $\pi^{*}$ as well as the minimum edit effect $\mathcal{E}^{\pi*}:=(\mathcal{E}_{N-}^{\pi*},\mathcal{E}_{E+}^{\pi*},\mathcal{E}_{E-}^{\pi*})$.
\end{theorem}
\begin{proof}
Without loss of generality, we set the cost $wt_C = 1, \forall C \in soft( \Phi_{{\satdiff}})$. Then our optimization objective is:
\begin{align*}
    \min_{\pi} cost(\pi,\Phi_{{\satdiff}}) &= \min_{\pi} \sum_{C|C \in soft( \Phi_{{\satdiff}}) \wedge  \pi \not \models C} 1 = \min_{\pi} \sum_{C|C \in (\Phi_{DE} \wedge \Phi_{AE}) \wedge  \pi \not \models C} 1 \\
    & = \min_{\pi} (\sum_{C|C \in \Phi_{DE}  \wedge  \pi \not \models C} 1 + \sum_{C|C \in \Phi_{AE} \wedge  \pi \not \models C} 1)\\ & = \min_{\Delta^{\pi}} (|\mathcal{E}_{E-}^{\pi}| +  |\mathcal{E}_{E+}^{\pi}|)=\min_{\mathcal{E}^{\pi}} |\mathcal{E}_E^{\pi}| \text{ 
 } \text{ 
 }\text{ 
 }\text{ } 
\end{align*}
% where
% \begin{align*}
% \Delta_{E}^{\pi} \circ E
% &=\begin{cases}
%     E \cup \{(n,e_i,n_c)\} &  \text{if } (n,e_i,n_c) \in \bigcup_{(n',e_i) \in E}  \widetilde{E}((n',e_i)),\text{ } \evar{n}{e_i}{n_c} = 1 \\
%     E \setminus \{(n,e_i,n_c)\} & \text{if } (n,e_i,n_c) \in E, \text{ }\evar{n}{e_i}{n_c} = 0 
%     \end{cases}
% \end{align*}
Intuitively, the solver finds a correct (by Theorem \ref{thm:sol_correct}) and optimal solution $\pi^{*}$ that minimizes the number of added potential edges and deleted existing edges, which is the minimum edit effect $\mathcal{E}^{\pi*}$. 

% And $\mathcal{E}_{N-}^{\pi}$ does not relate to the optimality aspect by Theorem \ref{thm:sol_correct}.

% The optimization is done and trusted by the solver with the above objective.

\end{proof}

\subsection{The decoding phase --- Decoding low-level edit scripts from {\maxsat} solution}
\label{sec:phase2}
% In this phase, we will demonstrate how a low-level edit script, composed of atomic edit actions, can be derived from the {\maxsat} solution. For a formal discussion of low-level edit scripts and atomic edit actions, we introduce the notation of $free$ edges and $free$ nodes.

% Conversely, if there is one child node connected to the parent node through edge $e$, we say that the edge is in a state of $\neg free$.

\begin{definition}[\emph{Free edge}]
An edge $e$ is considered to be $free$ if there is no child currently connected to a parent node through that edge. We denote the collection of edges $e$ that are $free$ in a graph $(N,E)$ as $E_{free}$.
\end{definition}

\begin{definition}[\emph{Free node}]
A node $n$ is considered to be $free$ if there is no edge connecting $n$ to a parent node. In other words, $n$ is the root of a subtree. We denote the collection of nodes $n$ that are $free$ in a graph $(N,E)$ as $N_{free}$.
\end{definition}

% Conversely, if there is an edge connecting $n$ to its parent node, we say that the node is $\neg free$.

\begin{definition}[\emph{Low-level edit actions}] We introduce three low-level (\emph{atomic}) edit actions: \texttt{dcon}, \texttt{del}, and \texttt{con}. These actions execute the node and edge operations described in edit effect $\mathcal{E}$. Specifically, \texttt{dcon} removes edges in $\mathcal{E}_{E-}^{\pi}$, \texttt{del} deletes nodes in $\mathcal{E}_{N-}^{\pi}$, and \texttt{con} adds edges in $\mathcal{E}_{E+}^{\pi}$, as illustrated in Figure \ref{fig:atomic}. We define each edit action as follows:
\begin{itemize}
    \item $\con{n_p}{i}{n}$: This action connects node $n$ to its new parent $n_p$ through the edge $e_i$, indicating that $n = n_p.c_i$. If node $n$ has children, this action also connects the subtree rooted at $n$ to the tree where $n_p$ belongs.
    \item $\dcon{n_p}{i}{n}$: This action detaches node $n$ from its parent $n_p$ by disconnecting edge $e_i$. After this action, the subtree rooted at node $n$ will be detached from the current tree and will become free, and so will the edge $(n_p,e_i)$. The disconnected subtree will be used in future actions.   
    \item $\del{n_p}{i}{n}$: This action deletes the subtree rooted at node $n$ from its parent $n_p$ by disconnecting edge $\fe{n_p}{i}$. The edge $\fe{n_p}{i}$ becomes free, and the subtree that was previously connected to it is removed and will not be used in future actions.
\end{itemize}
\end{definition}

% To illustrate, we demonstrate how to decode SAT solution examples for the tree inputs in Figure \ref{fig:overview_sat}.

\begin{figure}
    \centering

    \includegraphics[width=0.9\textwidth, trim={0 3.5cm 0 1cm},clip]{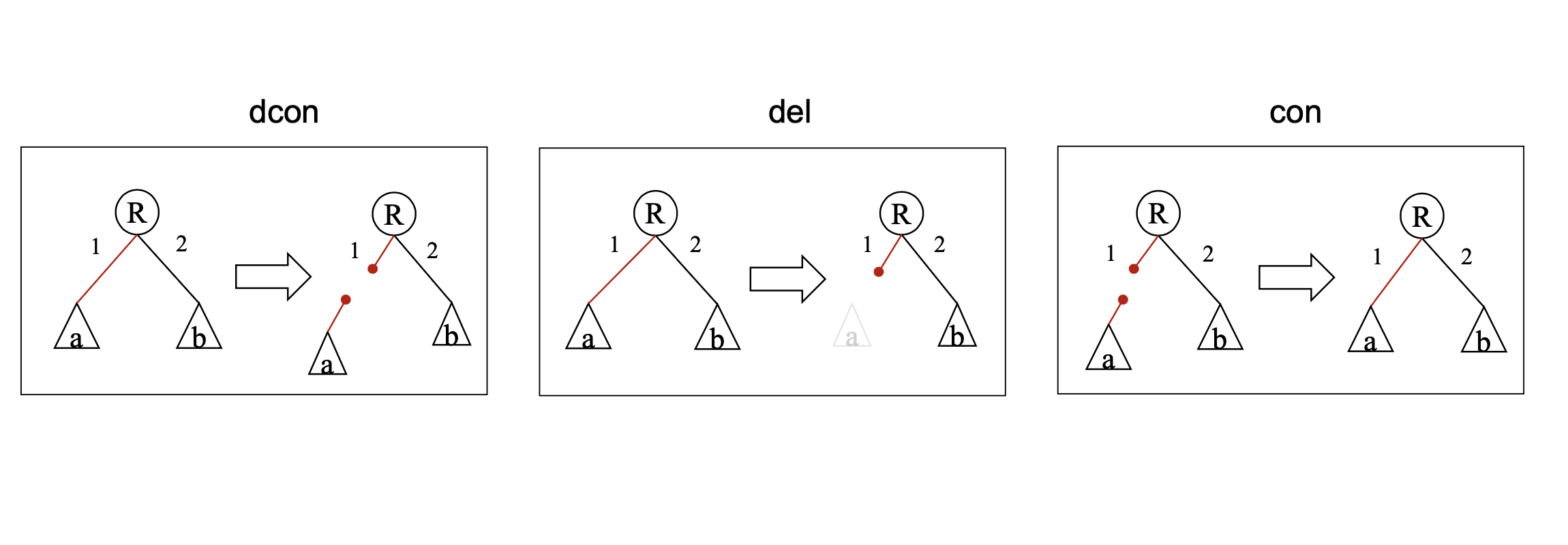}
    \caption{Low-level edit actions execute edge and node operations described in the edit effect.}
    \label{fig:atomic}
\end{figure}

\begin{remark}

In the earlier section, we demonstrate that SAT solutions can be interpreted as edit effects, which are executed by low-level edit actions. To illustrate, we provide examples of decoding SAT solutions into low-level edit actions based on the AST pairs introduced in Figure \ref{fig:overview_sat}.

\begin{itemize}
    \item $\evar{2}{1}{4} = 0$ and $\mvar{4}{18} = 1$ imply that the edge $\fe{2}{1}$ is available, and node 4 is needed later due to a matching literal, resulting in the atomic edit action $\dcon{2}{1}{4}$.
    \item In another scenario with $\evar{3}{2}{7} = 0$ and no true match variables for node 7, it suggests the deletion of the subtree rooted at node 7, leading to the atomic edit action $\del{3}{2}{7}$.
    \item $\evar{3}{2}{4} = 1$ indicates a connection between nodes 3 and 4 via the edge $\fe{3}{2}$, the derived atomic edit action is $\con{3}{2}{4}$.
\end{itemize}
\end{remark}

\begin{definition}[\emph{Edit script}]
An edit script, denoted as $\Delta$, is a sequence of atomic edit actions. Its semantics can be seen as a mapping defined on the graph space $G$. Formally, this mapping is represented as $ \llbracket \Delta \rrbracket  : G \rightarrow G{\perp}$, which yields ${\perp}$ if $\Delta (S') \ncong T$, where $S', T \in G$. 
And the correctness of $\Delta$ is indicated by $\Delta (S') \cong T$. 
% The edit script $\Delta$ is compositional. Formally, we have: $\Delta :=  \delta_1 | \delta_2 | ... | \delta_n | \Delta_1; \Delta_2$, where $\{\delta_i | i \in \{1,2,...,n\}\}$ is a sequence of atomic edit actions, and $\Delta_1$ and $\Delta_2$ are  subsequences of edit scripts $\Delta$.

% It can be formed by simply combining smaller edit scripts, i.e., $\Delta:=(\Delta_1, \Delta_2, ..., \Delta_n)$, with the standard semantics: $\llbracket \Delta_1, \Delta_2, ..., \Delta_n \rrbracket = \llbracket \Delta_1 \rrbracket \circ \llbracket \Delta_2 \rrbracket \circ ... \circ \llbracket \Delta_n \rrbracket$. The minimum edit scripts that cannot be further broken down are called edit actions/operations.

\end{definition}

\begin{definition}[\emph{Low-level edit script}] A \emph{low-level edit script}, denoted as $\underline{\Delta}$, is an edit script composed of low-level edit actions: $\con{n_p}{i}{n}$, $\dcon{n_p}{i}{n}$, and $\del{n_p}{i}{n}$.  We present the syntax of $\underline{\Delta}$ as follows: $$\underline{\Delta} := \ \ \cdot \ \ | \ \  \texttt{con} \ \ | \ \ \texttt{dcon} \ \ | \ \ \texttt{del} \ \  | \  \ \underline{\Delta}_1; \underline{\Delta}_2$$
 
Note that $\underline{\Delta}$ must be arranged in an appropriate execution order. An execution order is considered \emph{appropriate} when there is no conflict, meaning that an edge or node is used only when it is in the free state. Then we present the semantics of $\underline{\Delta}$ as follows:

% \vspace{10 pt}
\begin{prooftree}
    \AxiomC{}
    % \RightLabel{(No-op)}
    \UnaryInfC{ $\rho_0 \vdash_{L}  \cdot  \Downarrow \rho_0 $ }
    \DisplayProof \ \ \ \ \ 
    \AxiomC{$\rho_0 \vdash_{L}   \underline{\Delta}_1  \Downarrow \rho_1 $}
    \AxiomC{$\rho_1 \vdash_{L}   \underline{\Delta}_2  \Downarrow  \rho_2 $}
    % \RightLabel{xxx}
    \BinaryInfC{ $\rho_0  \vdash_{L}   \underline{\Delta}_1;\underline{\Delta}_2  \Downarrow  \rho_2 $ }
\end{prooftree}
    
% \AxiomC{$\rho \vdash_{L} n\ is\ free$}
\begin{prooftree}
    \AxiomC{$n \in F_0 $}
    \AxiomC{$(n_p, e_i) \in F_0 $}
    \AxiomC{$(n_p, e_i, n) \notin S_0'$ }
    \TrinaryInfC{ $\rho_0  \vdash_{L} \texttt{con}(n_p, e_i, n) \Downarrow S_0' \cup \{(n_p, e_i, n)\} | F_0 \setminus \{n, (n_p,e_i) \} $ }
\end{prooftree}

% \AxiomC{$\rho \vdash_{L} n\ is\ \neg free$}
\begin{prooftree}
    \AxiomC{$ n \notin F_0 $}
    \AxiomC{$ (n_p, e_i) \notin F_0 $}
    \AxiomC{$ (n_p, e_i, n) \in S_0' $}
    \TrinaryInfC{ $\rho_0  \vdash_{L} \texttt{dcon}(n_p, e_i, n) \Downarrow  S_0' \setminus \{(n_p, e_i, n)\} | F_0 \cup \{n, (n_p,e_i) \} $ }
\end{prooftree}

% \jl{verify}
\begin{prooftree}
    \AxiomC{$\rho_0  \vdash_{L} \texttt{dcon}(n_p, e_i, n) \Downarrow \rho_1 $}
    \AxiomC{$\rho_1 \vdash_{L} \texttt{del}(n,e_1,n._{c_1}) \Downarrow \rho_2 $}
    \AxiomC{$\rho_2 \vdash_{L} \texttt{del}(n,e_2,n._{c_2}) \Downarrow \rho_3 $}
    \TrinaryInfC{$\rho \vdash_{L} \texttt{del}(n_p, e_i, n) \Downarrow  S'_3 \setminus \{n\} |  F_3 $}
\end{prooftree}

% \begin{prooftree}
%     \AxiomC{$ n: free \rightarrow \neg free, (n_p,e_i): free \rightarrow \neg free, E \rightarrow E \cup \{(n_p, e_i, n)\}$}
%     \UnaryInfC{$T \rightarrow connect(n, n_p, e_i) \circ T$}
% \end{prooftree}

% \begin{prooftree}
%     \AxiomC{$ n: \neg free \rightarrow free, (n_p,e_i): \neg free \rightarrow free, E \rightarrow E \setminus \{(n_p, e_i, n)\}$}
%     \UnaryInfC{$T \rightarrow disconnect(n, n_p, e_i) \circ T $}
% \end{prooftree}

% \begin{prooftree}
%     \AxiomC{$ (n_p,e_i): \neg free \rightarrow free, N \rightarrow N \setminus \{n\}$}
%     \UnaryInfC{$T \rightarrow delete(n, n_p, e_i)\circ T$}
% \end{prooftree}
\vspace{10 pt}
where $ \rho_i := S'_i|F_i $, $S'_i$ is the current source space and $F_i$ is set of \emph{free} nodes and \emph{free} edges in the current source space $S'_i$,  $i \in \{0,1,2,3\}$. $\underline{\Delta}_1$ and $\underline{\Delta}_2$ are any two sequences of low-level edit scripts. We say $\underline{\Delta}$ is correct if, after the execution of $\underline{\Delta}$, we have:  
$S' | N_T \vdash_L \underline{\Delta}(S') \Downarrow T | \{\ \}$.

\end{definition}

\begin{theorem}
\label{thm:opt_correct_low}
[\emph{Correct and minimum low-level edit scripts}] Any correct solution $\pi$ to $\Phi_{{\satdiff}}$ can be decoded into a correct low-level edit script $\underline{\Delta}_{\pi}$. The optimal solution $\pi^*$ yields the correct and minimum low-level edit script $\underline{\Delta}^*_{\pi}$. 
\end{theorem}

\begin{proof}
It is trivial to show that a low-level edit script $\underline{\Delta}_{\pi}$ corresponds to an edit effect $\mathcal{E}^{\pi}$ associated with an appropriate execution order. By Theorem \ref{thm:sol_correct} and \ref{thm:sol_opt}, the above theorem holds. It only remains to show that an appropriate execution order must exist when $\mathcal{E}^{\pi}$ is correct. 

Note that $\mathcal{E}^{\pi}$ being correct implies that no free edges and nodes are left when execution finishes. Since each low-level edit action either keeps the free state unchanged or changes the state to its opposite, then an appropriate execution ordering must exist. It can be constructed by taking into account the dependencies between edit actions, as presented in Algorithm \ref{alg:get_dependency_graph}.

\end{proof}

\subsection{The synthesis phase --- Merging high-level edit actions from low-level edit actions}
\label{sec:phase3}

% For the same reason that type safety is crucial in low-level edit scripts, we introduce a new concept, referred to as \emph{Bookkeeping}, to ensure type safety of the high-level operation, \emph{Update}. \emph{Update},as an operation that updates a single node, consists of one deletion and one connection of virtual node. \texttt{del}, by its definition, deletes the whole sub-tree, combining with a connection at the same place means to "update" a sub-tree to a virtual node. This conflicts with the definition of \emph{Update}. To resolve this issue, we utilize Child Management to not only make it type safe, but also enable the synthesizing phase to capture structural updates.

% \jl{definition of update and typing rule}
While low-level edit actions characterize the necessary transformation from the source tree to the target tree, they are not intuitive. To improve readability for end users, we introduce five high-level edit actions.
\begin{definition}[\emph{High-level edit actions}] 
\label{def:high-level-edits}
We define the following high-level edit actions:
\begin{itemize}
    \item $\Ins{n_p}{i}{n}$: This operation involves adding a \emph{virtual node} $n$ to the $i$\textsuperscript{th} child position of node $n_p$ through $e_i$. 
    
    \item $\Del{n_p}{i}{n}$:  This action deletes the subtree rooted at node $n$ from its parent $n_p$ by disconnecting edge $e_i$. The deleted subtree will not be used again. 

    \item $\Upd{n}{l}{v}$: 
    This editing action involves updating the label and/or value of node $n$ in the tree. It replaces the current label with a new assignment $l$ and/or updates the value with $v$. Importantly, this action maintains the node's position within the tree structure. 
     
    \item $\Mov{n}{n_p}{i}{n_p'}{j}$: 
    This operation involves moving node $n$ along with its subtree from being the $i$\textsuperscript{th}  child of its current parent $n_p$ to becoming the $j$\textsuperscript{th}  child of its new parent $n_p'$. This means that both the node $n$ and the subtree rooted at $n$ are relocated to the new position.

    \item $\Swp{n}{n_p}{i}{n'}{n_p'}{j}$: This operation consists of two $\texttt{Mov}$'s whose destination conflicts with the other's parent node. However, $\texttt{Swp}$ merges the low-level actions of two $\texttt{Mov}'s$ differently to avoid this conflict. Furthermore, our future extension of this project $\texttt{Swp}$ can generalize to an unlimited amount of conflicting $\texttt{Mov}$'s.
    
    % \xs{Optional: further mention that two moves can be composed into a swap edit. More high-level edits like this are our future extensions}
\end{itemize}
\end{definition}

\begin{definition}[\emph{High-level edit script}]
\label{def:sem_high}
A high-level edit script, denoted as $\overline{\Delta}$, is an edit script composed of high-level edit actions including \texttt{Upd, Mov, Del, Ins} and \texttt{Swp}. We present the syntax of $\overline{\Delta}$ as follows: 
$$\overline{\Delta} ::=  \ \ \cdot \ \ | \ \ \texttt{Ins} \ \ | \ \ \texttt{Del} \ \ | \ \ \texttt{Upd} \ \ | \  \ \texttt{Mov} \ \ | \ \  \texttt{Swp} \ \ | \ \ \overline{\Delta}_1; \overline{\Delta}_2$$
These high-level edit actions can be merged by low-level edit actions and require the same condition for an \emph{appropriate} execution order, i.e., an edge or node is used only when it is in the free state. We present the semantics of $\overline{\Delta}$ as follows:

\begin{prooftree}
    \AxiomC{}
    % \RightLabel{(No-op)}
    \UnaryInfC{ $\rho_0  \vdash_{H}  \cdot  \Downarrow \rho_0  $ }
    \DisplayProof \ \ \ \ \ 
    \AxiomC{$\rho_0  \vdash_{H}   \overline{\Delta}_1  \Downarrow \rho_1 $}
    \AxiomC{$\rho_1  \vdash_{H}   \overline{\Delta}_2  \Downarrow  \rho_2 $}
    \BinaryInfC{ $\rho_0  \vdash_{H}   \overline{\Delta}_1;\overline{\Delta}_2  \Downarrow  \rho_2 $ }
    \DisplayProof \ \ \ \ \
    \AxiomC{$\rho_0  \vdash_{L} \del{n_p}{i}{n} \Downarrow \rho_1 $}
    \UnaryInfC{$\rho_0  \vdash_{H} 
    \Del{n_p}{i}{n} \Downarrow  \rho_1 $}
\end{prooftree}

\begin{prooftree}
    \AxiomC{$ n'  \in F_0 $}
    \AxiomC{$\rho_0 \vdash_{L} \texttt{del}(n_p, e_i, n) \Downarrow \rho_1 $}
    \AxiomC{$\rho_1 \vdash_{L} \texttt{con}(n_p, e_i, n') \Downarrow \rho_2 $}
    \TrinaryInfC{ $\rho_0 \vdash_{H} \Upd{n}{l}{v} \Downarrow  \rho_2 $ }
\end{prooftree}

\begin{prooftree}
    \AxiomC{$ n \in F_0 $}
    \AxiomC{$ n \in S'_0 $}
    \AxiomC{$\rho_1 \vdash_{L} \texttt{con}(n_p, e_i, n) \Downarrow \rho_2 $}
    \TrinaryInfC{ $\rho_0 \vdash_{H} \texttt{Ins}(n,n_p,e_i) \Downarrow  S'_2| F_2 \cup \{ (n,e_1),(n, e_2) \} $} 
\end{prooftree}

\begin{prooftree}
    \AxiomC{$ (n_p', e_j)  \in F_0 $} 
    \AxiomC{$\rho_0 \vdash_{L} \texttt{dcon}(n_p, e_i, n) \Downarrow \rho_1 $}
    \AxiomC{$\rho_1 \vdash_{L} \texttt{con}(n_p', e_j, n) \Downarrow \rho_2 $}
    \TrinaryInfC{ $\rho_0 \vdash \Mov{n}{n_p}{i}{n_p'}{j} 
    \Downarrow  \rho_2 $ }
\end{prooftree}

\begin{prooftree}
    % \AxiomC{$\rho  \ is\ free$} 
    \AxiomC{$\rho_0 \vdash_{L} \texttt{dcon}(n_p', e_j, n') \Downarrow \rho_1 $}
    \AxiomC{$\rho_1 \vdash_{H} \Mov{n}{n_p}{i}{n_p'}{j} 
    \Downarrow  \rho_2 $}
    \AxiomC{$\rho_2 \vdash_{L} \texttt{con}(n_p, e_i, n) \Downarrow \rho_3 $}
    % \AxiomC{$\rho' \vdash_{L} \texttt{con}(n_p', e_j, n) \Downarrow \rho'' $}
    \TrinaryInfC{ $\rho_0 \vdash_{H}   \Swp{n}{n_p}{i}{n'}{n_p'}{j} \Downarrow  \rho_3 $ }
\end{prooftree}

\end{definition}

To determine the appropriate execution order for both low-level and high-level edit scripts, we utilize a dependency graph to characterize the complex dependencies between edit actions and node relationships. We define the dependency graph as follows.

\begin{definition}[\emph{The dependency graph}]
A dependency graph, denoted as $K$, is a directed acyclic graph (DAG) with vertices representing low-level edit actions and directed edges indicating dependencies. More specifically, for two distinct edit actions $u$ and $w$, we connect them with the edge $u \rightarrow w$. The notation $u \rightarrow w$ signifies that $u$ depends on $w$, indicating that $w$ must occur before $u$.
\end{definition}
% The dependency graph tpological ordering of
% We define the 
% For each set of edit actions that we receive from the Max-SAT solver, a corresponding dependency graph is necessary for deducing a type safe edit script.
% We use the notation \text{"}A $\rightharpoonup$ B\text{"} to refer the notion \text{"}$A$ depends on action $B$\text{"} or in other words, $B$ has to happen before $A$. Then if we recall the definition of our low-level edit actions, their dependency relations are described as follows:
% \begin{itemize}
%     \item $\con{n_p}{i}{n}$ \textbf{$\rightharpoonup$} $\con{n_{gp}}{j}{n_p}$ or $\dcon{n_{p}}{i}{n'}$ or $\del{n_{p}}{i}{n'}$.
%     \item $\del{n_{p}}{i}{n}$ $\rightharpoonup$ $\forall$ 
%     \item $\dcon{n_p}{i}{n}$ has the highest priority, i.e., no action must happen before \texttt{dcon}.
% \end{itemize}
\vspace{-10pt}
\begin{algorithm2e}[]
    \caption{Construct Dependency Graph}
    \label{alg:get_dependency_graph}
    \small
    \DontPrintSemicolon
    % \SetKwInOut{KwIn}{myinput }
    % \SetKwInOut{Param}{Parameterxxxxx}
    \SetKwProg{Fn}{Function}{}{end}
    \SetKwFunction{getDPG}{GetDependencyGraph}
    % \KwIn{a graph, $dpg$, whose edge set E($dpg$)=$\emptyset$ and vertex set V($dpg$)=SAT solution}
    \KwIn{a set of low-level edit actions, $V$}
    \KwOut{a dependency graph of low-level edit actions, $\langle V,E\rangle$}
    % \Fn{ \getDPG{$dpg$} }{ 
        $E \leftarrow \emptyset$ \;
        \For{$action$ $\textbf{in}$ $V$}{
            \Switch{$action$}{
                \Case{$\con{n_p}{i}{n}$}{
                $E$ $\leftarrow$ $E$ $\cup$ \{$action \rightarrow \con{n_{gp}}{j}{n_p} $ | $n_{gp}$ is the parent of $n_{p}$, and  $n_{p}$ is the parent of $n$\}\;
                $E$ $\leftarrow$ $E$ $\cup$ \{$action \rightarrow \dcon{n_{p}}{i}{n'}$ or $\dcon{n_{p}'}{j}{n}$\}\;
                $E$ $\leftarrow$ $E$ $\cup$ \{$action \rightarrow \del{n_{p}}{i}{n'} $ or $\del{n_{p}'}{j}{n}$\}\;
                }
                \Case{$\del{n_p}{i}{n}$}{ $E$ $\leftarrow$ $E$ $\cup$ \{$action \rightarrow \dcon{n_{p}'}{j}{n'}$ | $n_p'$ is a descendant of $n$\}\;
                }
            }
        % }
}      
        \Return{$\langle V, E\rangle$}
\end{algorithm2e} 
\vspace{-10pt}
                
            %     $S$ $\leftarrow$ \{$\con{n_{gp}}{j}{n_p}$ or $\dcon{n_{p}}{i}{n'}$ or $\del{n_{p}}{i}{n'}$ $\in$ $V(dpg)$\}\;
            % $E(dpg)$ $\cup$ \{($op \rightarrow op'$) | $op' \in S$ \}\;

\begin{remark} [Determining an appropriate order for edit scripts]
Given a set of low-level edit scripts, we run Algorithm \ref{alg:get_dependency_graph} to get the dependency graph. Then we can construct an appropriate order as follows: For each connected component, we generate a sequence by recursively slicing out each independent node. Specifically, if there is a dependency $u \rightarrow w$ between nodes $u$ and $w$, meaning $u$ depends on $wu$, slicing out $w$ makes $u$ independent. This process results in a sequence of edit actions for each connected component. These sequences can then be merged in any order as they are independent from each other. \\ 
To better illustrate this process, we use the ASTs from Figure \ref{fig:overview_sat1} and generate the corresponding dependency graph in Figure \ref{fig:topo_example}. Following the aforementioned procedure, there are two connected components. For the one shown in Figure \ref{fig:topo_a}, we first slice out independent vertices $\del{3}{2}{7}$ and $\dcon{2}{1}{4}$, then $\con{3}{2}{4}$ becomes independent, so we slice it out next. Then we end up with a sequence $(\del{3}{2}{7};\dcon{2}{1}{4};\con{3}{2}{4};\con{2}{1}{21})$. For the other connected component shown in Figure \ref{fig:topo_b}, we first slice out $\dcon{5}{1}{8}$ and $\dcon{5}{2}{9}$, then $\del{2}{2}{5}$ becomes independent and needs to be sliced out next. Repeated the process, we have a sequence $(\dcon{5}{1}{8};\dcon{5}{2}{9};\del{2}{2}{5}; \con{2}{2}{22};\con{22}{1}{8};\dcon{22}{2}{9})$. The final low-level edit scripts can be merged using these two sequences in any order. To show the execution order is appropriate, we present detailed tracking of free edges and nodes in Table \ref{tab:atomic_ordering}. 
\end{remark}

% \begin{figure}
%     \centering    \includegraphics[width=0.75\textwidth,page=3, trim={0 3cm 0 3.5cm}, clip]{figures/topo.pdf} 
%     \caption{A code edit against user's intuition.}
%     \label{fig:against_intui}
% \end{figure}

\begin{figure}
    \subfloat[Ins, Mov and Del]{
        \centering
        \includegraphics[width=0.29\textwidth, page=1, trim={0 2cm 18cm 2cm},clip]{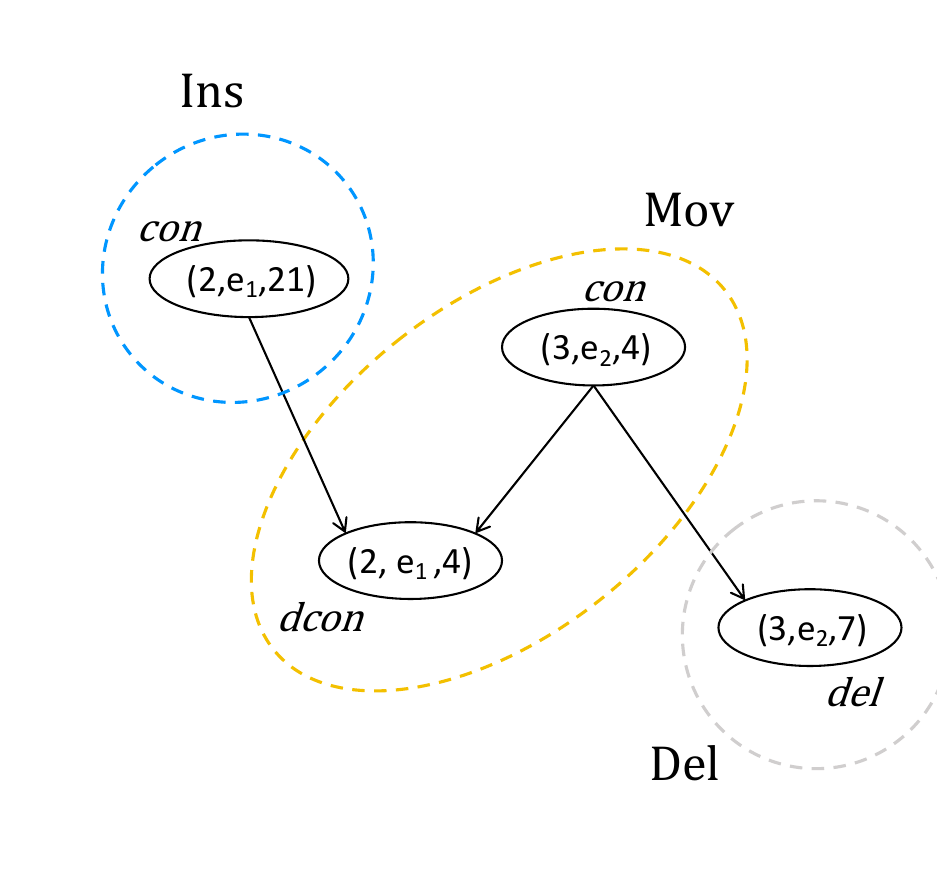}
        \label{fig:topo_a}
    }
    \subfloat[Upd]{
    \centering
        \includegraphics[width=0.29\textwidth, page=1, trim={16cm 2cm 2cm 2cm},clip]{figures/topo.pdf}
        \label{fig:topo_b}
    }
    \subfloat[Swp]{\includegraphics[width=0.33\textwidth, page=2, trim={0cm 2cm 16cm 2cm},clip]{figures/topo.pdf}
    \label{fig:topo_c}
    }
    \caption{Dependency graph of low-level edit actions. The dashed regions contain a group of low-level edit actions that can merged to high-level edit actions.}
    \label{fig:topo_example}
    \vspace{-5pt}
\end{figure}

\vspace{-7pt}

\begin{remark} [Synthesize high-level edit scripts]
Once we obtain the dependency graph of low-level edit actions, we can determine the appropriate execution order for both low-level and high-level edit scripts. Then it remains to merge certain subsequences/groups of low-level edit actions into high-level operations. To achieve this, we introduce the following templates:

\begin{itemize}
    \item $\Mov{n}{n_p}{i}{n_p'}{j}$: $\con{n_p'}{j}{n}$ $\rightarrow$ $\dcon{n_p}{i}{n}$
    \item $\Swp{n}{n_p}{i}{n'}{n_p'}{j}$: $\Mov{n}{n_p}{i}{n_p'}{j}$ $\rightleftarrows$ $\Mov{n'}{n_p'}{j}{n_p}{i}$
    \item $\Upd{n}{l}{v}$: $\con{n_p}{l}{v}$ $\rightarrow$ $\del{n_p}{l}{n}$ $\rightarrow$ $\{\dcon{n}{i}{n_c} | \con{v}{i}{c} \rightarrow \con{n_p}{l}{v} \}$
    \item $\Ins{n_p}{i}{n}$/$\Del{n_p}{i}{n}$: independent $\texttt{del}$/$\texttt{con}$ are not involved in the other templates.
\end{itemize}

As declared in Definition \ref{def:high-level-edits} that a $\Mov{n}{n_p}{i}{n_p'}{j}$ is made of $\con{n_{p}'}{j}{n}$ and $\dcon{n_p}{i}{n}$, the $\texttt{Mov}$ template simply specifies the dependency relation between the two actions. Intuitively, to move a subtree rooted at $n$, we need to disconnect $n$ from its parent $n_p$ first so that $n$ is ready to be moved, thus $\con{n_{p}'}{j}{n}$ $\rightarrow$ $\dcon{n_p}{i}{n}$. Next, the $\texttt{Swp}$ operation means two nodes switching their children. Notably, $\texttt{Swp}$ requires two $\texttt{Mov}$ to depend on each other, i.e., $\texttt{Mov1}$ $\rightleftarrows$ $\texttt{Mov2}$, so neither of them can be sliced out using the $\texttt{Mov}$ template. In other words, having $\texttt{Swp}$ avoids the awkward non-halting situation of our synthesizing algorithm for high-level edit script. \texttt{Upd} on a node essentially involves disconnecting it from its children, deleting it, and filling the hole with a new node, so that we can connect the children back.
Thus, $\con{n_p}{l}{v}$ depends on $\del{n_p}{l}{n}$, which in turn depends on $\dcon{n}{i}{n_c}$ for all children. In addition, the templates for $\texttt{Ins}$ and $\texttt{Del}$ are simply the $\texttt{con}$ and $\texttt{del}$ that cannot fit in any of the above templates.
To give some concrete examples, we show how to synthesize high-level edit scripts based on tree inputs shown in Figure \ref{fig:overview_sat}.
\begin{itemize}
    \item We first slice out $\del{3}{2}{7}$ as it is independent. Since it does not match any template, it is recognized as $\Del{3}{2}{7}$.
    \item The action$\con{3}{2}{4}$ depends on $\dcon{2}{1}{4}$, which matches our template for \texttt{Mov}. So we merge them into $\Mov{4}{2}{1}{3}{2}$ and then slice out.
    \item Next, the action $\con{2}{2}{21}$ becomes independent. Since it does not match any template, it can be recognized as an insert action $\Ins{2}{1}{21}$.
    \item The action $\del{2}{2}{5}$ depends on $\dcon{5}{1}{8}$ and $\dcon{5}{2}{9}$. and 
    $\con{2}{2}{22}$ depends on it. Moreover, $\con{22}{1}{8}$ and $\con{22}{2}{9}$ depends on $\con{2}{2}{22}$, then these actions match $\texttt{Upd}$ template. So we merge them into $\texttt{Upd}(5,Add)$ and then slice out.
\end{itemize}

\end{remark}

\begin{theorem} [\emph{Guarantee of correctness}] 
\label{thm:type-safe-low-edit} 
The high-level edit script $\overline{\Delta}$, synthesized from a correct low-level edit script $\underline{\Delta}$ using dependency graph $K$, is also correct.
\end{theorem}

\begin{proof}
We provide a proof sketch. 
Since the low-level edit script $\underline{\Delta}$ has the appropriate execution order determined by the dependency graph $K$ generated by Algorithm \ref{alg:get_dependency_graph}, the corresponding high-level edit script $\overline{\Delta}$ also exhibits the same appropriate execution order. Then merging low-level edit actions into high-level edit actions has no impact on the edit effect. Thus, the synthesized $\overline{\Delta}$ must be correct when $\underline{\Delta}$ is correct.
\end{proof}

\begin{table}
\begin{tabular}{lccc}
  \hline
      High-level edits ($\overline{\Delta}$)  & Low-level edits ($\underline{\Delta}$) & \textit{Free} Node ($N_{free}$)  & \textit{Free} Edge ($E_{free}$)  \\
  \hline
  \hline
  
   &  & 21,22 &  $\emptyset$  \\
  \hline
   $\Del{3}{2}{7}$ & $\del{3}{2}{7}$ &   21,22 & $\fe{3}{2}$   \\
  \hline
  
    \multirow{2}{*}{$\Mov{4}{2}{1}{3}{2}$} &$\dcon{2}{1}{4}$  & 21,22,4 & $\fe{2}{1}$, $\fe{3}{2}$   \\
    & $\con{3}{2}{4}$ & 21,22 & $\fe{2}{1}$ \\
    \hline
  $\Ins{2}{1}{21}$ & $\con{2}{1}{21}$ & 22 & $\emptyset $\\
    \hline
  \multirow{6}{*}{$\texttt{Upd}(5,Add)$} & $\dcon{5}{1}{8}$ & 22,8 & $\fe{5}{1}$  \\
  & $\dcon{5}{2}{9}$ & 22,8,9 & $\fe{5}{1}, \fe{5}{2}$ \\
  & $\del{2}{2}{5}$ & 22,8,9 & $\fe{2}{2}$\\
  & $\con{2}{2}{22}$ & 8,9 & $\fe{22}{1},\fe{22}{2}$ \\
  & $\con{22}{1}{8}$ & 9 & $\fe{22}{2}$ \\
  & $\con{22}{2}{9}$ & $\emptyset$ & $\emptyset$ \\
  \hline
\end{tabular}
\caption{
The low-level edit script, which comprises atomic edit actions and detailed tracking of nodes and edges, ensures correctness and type safety. The first row indicates the initial state of free nodes and edges.}
\label{tab:atomic_ordering}
\vspace{-20pt}
\end{table}

\vspace{-7.5pt}

\subsection{Extensions}

\paragraph{Supporting type safety}

\cite{truediff} highlight a significant drawback of \textit{Gumtree} and many other popular tree diffing algorithms: the lack of type safety. Specifically, they generate ill-typed intermediate trees that can only be captured by an untyped tree representation. To address this issue, \textit{truediff} provides assurances that all intermediate trees are well-typed.
Importantly, our edit script can achieve the same level of type safety. For a detailed discussion on the type systems, along with the type safety guarantees of our edit scripts, we refer readers to Appendix \ref{sec:type}.

% {\satdiff} can be easily extended to support type safety if type information is given. 
% During the encoding phase, type information can be enrolled to 

% \begin{wrapfigure}[6]{r}{10cm}
% \hspace{3.5cm}
% \begin{minipage}{0.3\textwidth}
%     \centering   
%     \begin{lstlisting}[language=Python,basicstyle=\ttfamily\footnotesize] 
% def fib((*@\colorbox{yellow}{\textcolor{red}{n = 1}}@*)): # source code
%     if n == 1:
%         (*@\colorbox{yellow}{\textcolor{red}{return  n}}@*)
%     else:
%         return n * fib(n-1)
% \end{lstlisting}
% \end{minipage}\hfill
% \begin{minipage}{0.3\textwidth}
% \centering
% \begin{lstlisting}[language=Python,basicstyle=\ttfamily\footnotesize] 
% def fib((*@\colorbox{yellow}{\textcolor{red}{n}}@*)): # target code
%     if n == 1:
%         (*@\colorbox{yellow}{\textcolor{red}{return 1}}@*)
%     else:
%         return n * fib(n-1)
%     \end{lstlisting}
% \end{minipage}
% \caption{A code edit against user's intuition}
% \label{fig: Example program}
% \end{wrapfigure}

% \allen{Finetune the weights to align with human's intention, can be done through a user study, validity }

\paragraph{Supporting flexible edits}
% In most cases, minimum low-level edit scripts are subsequently synthesized into minimum high-level edit scripts. However, it's interesting to note that in specific situations, a minimum low-level edit script may not necessarily translate into a minimum high-level edit script. This can happen because different high-level edit actions may consume a varying number of atomic edit actions. 

\satdiff \ by default assigns each low-level edit the same weight, which can be customized to prioritize different kinds of edits. The default weights may not always align with the programmers' intention. One common scenario is when the edit script involves moving a frequently occurring constant such as 0, 1, or null, while the programmer may find it convenient to delete and retype these constants. To respect this edit habit, we can simply lower the weights of connecting virtual nodes representing these constants.

\section{Experimental Evaluation}
\label{sec:eval}

In this section, we aim to provide a comprehensive evaluation of $\satdiff$ and its capability to uncover high-quality edit scripts within the context of various datasets encompassing different programming languages.
We evaluate $\satdiff $ by comparing it with three state-of-the-art tree diffing algorithms on two datasets from the real world. 

\vspace{-5pt}

\paragraph{Datasets.}
We collect two large datasets used in the prior work~\cite{truediff,nate}. The first dataset consists of approximately 649 real-world Python files (around 1017 lines of code per file) extracted from the latest 500 commits of the popular deep learning framework, \emph{Keras}\footnote{https://keras.io/}, a high-level API for TensorFlow~\cite{tensorflow}. We filter out \textit{simple} commits that only add new code, which we believe is less interesting as the editing is just insertions.
The second dataset consists of 1139 students' functional programming homework snapshots written in OCaml, collected from the NATE project~\cite{nate}.

\vspace{-5pt}

\paragraph{Baselines.}
We compare $\satdiff$ against the latest version of the classic tree difference algorithm \textit{Gumtree}~\cite{gumtree} and its recent follow-up work, \textit{truediff}~\cite{truediff}, which provides type-safe guarantees, on the Python dataset. On the OCaml dataset, we compare $\satdiff$ with NATE's tree diffing algorithm, a generalization of Levenstein distance from strings to trees \footnote{We note that evaluating all baselines on both datasets is unfortunately not an option, due to their limited support for different languages.}.

\vspace{-5pt}

\paragraph{Setup.}
% We run two sets of experiments: automatic and manual evaluation. For automatic evaluation, we evaluate and compare {\satdiff} and its counterpart in terms of conciseness and runtime on each pair of source and target code. While conciseness is considered a good measure for edit scripts, it may have limitations in reflecting alignment with a programmer's intention. Thus, our manual evaluation requires three raters to independently assess the quality of the generated edit scripts of {\satdiff} and its counterpart  by answering the following questions:

% \begin{itemize}
%     \item \emph{Question \#1: Does {\satdiff} do a good job?}  \\
%           And the possible answers are: 
%           \begin{enumerate}
%               \item {\satdiff} does a good job: it helps to understand
% the code change.
%               \item Neutral.
%               \item {\satdiff} does a bad job: it doesn't help to understand
% the code change.
%           \end{enumerate} 
%     \vspace{5pt}
%     \item \emph{Question \#2:  Is {\satdiff}'s edit script  better than that of its counterpart? } \\
%     And the possible answers are:
%           \begin{enumerate}
%               \item {\satdiff} is better.
%               \item Equivalent.
%               \item Its counterpart is better. 
%           \end{enumerate}          
% \end{itemize}

All the experiments in this section are conducted on an Ubuntu22.04LTS machine equipped with 32 GB of RAM and a 2.7GHz processor. The {\maxsat} solver used in {\satdiff} is CASHWMaxSAT~\cite{lei2021cashwmaxsat}, the winner (weighted complete track) of the 2021 MaxSAT Competition. We set the timeout for the solver at 60 seconds. Note that $\satdiff$ still yields a correct edit script at the timeout, despite lacking minimal guarantees. Finally, to conduct a fair comparison, we count a swap operation as two move operations since it is not supported in our chosen baselines.

\subsection{Evaluation on conciseness}
% This dataset is a subset of the Keras dataset that \texttt{truediff} employs for evaluation as we exclude instances whose edit plans consist of solely insertions for constructing unused global expressions. Our focus lies on the more intriguing data that involves the need to adjust existing expressions. \\
% By examining these cases, we aim to assess whether $\satdiff$ can generate smaller type-safe edit plans compared to \texttt{Gumtree} and \emph{truediff} while maintaining a reasonable runtime.

\begin{figure}[th]
    \centering
    \label{fig:patch_size}
    \subfloat[]{
    \centering            \includegraphics[width=0.33\textwidth]{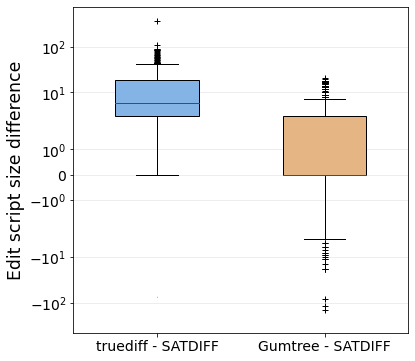}
            \label{fig:keras-diff}
        }
    \subfloat[]{
    \centering
            \includegraphics[width=0.32\textwidth]{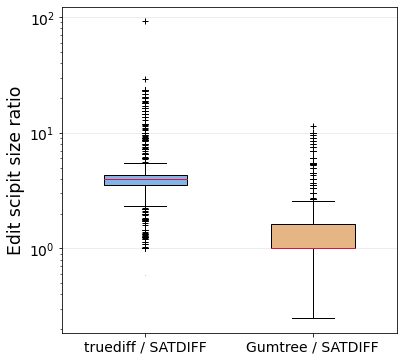}
            \label{fig:keras-div}
        } 
    \subfloat[]{
    \centering
    \includegraphics[width=0.32\textwidth]{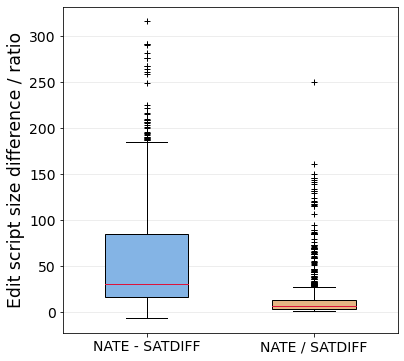}
    \label{fig:nate-vs-satdiff}
    }

    \caption{Absolute differences and relative differences in edit script sizes between different baselines and $\satdiff$. In Fig (a) we calculate the difference between the edit script size of the baseline approaches and \satdiff \ on the Python dataset while in Fig (b) we calculate their ratio. In Fig (c) we present both the difference and ratio on the OCaml dataset.}
    \vspace{-5pt}
\end{figure}

To assess and compare the conciseness of generated edit scripts, we compute edit scripts using baseline approaches and $\satdiff$ on each code file and then record their absolute difference ($|\Delta_{baseline}| - |\Delta_{\satdiff}|$) and the relative difference ($|\Delta_{baseline}| / |\Delta_{\satdiff}|$). We present the statistics of results using boxplots for both Python and OCaml benchmarks.

\vspace{-5pt}

\paragraph{Conciseness on Python benchmark}
The absolute difference in edit script sizes between baselines and $\satdiff$ is shown in Figure~\ref{fig:keras-diff}, in which the region above zero means that baselines require more edits than $\satdiff$. As we can observe, $\satdiff$ consistently outperforms \textit{truediff} across the entire dataset and is almost always superior to or equal to \textit{Gumtree}, except for a few outliers. Upon closer inspection, those outliers correspond to instances where $\satdiff$ timeout. In such scenarios, the solver returns the current best correct solution, which often deviates significantly from minimality. Figure~\ref{fig:keras-div} shows the relative differences in edit script sizes. 
We observe that, on average, \textit{truediff} requires $5.16\times$ more edits compared to $\satdiff$, while \textit{Gumtree} requires $1.72\times$  more edits. However, in specific instances, both \textit{truediff} and \textit{Gumtree} may require more than $10\times$ and $4\times$ the edits of $\satdiff$, respectively. We carefully examine these cases and discover that \textit{truediff} tends to generate a significant number of deletions and/or insertions, whereas $\satdiff$ effectively avoids them by employing flexible move operations. In scenarios involving the deletion of large chunks of code, such as a subtree, \textit{Gumtree} generates numerous deletion actions, while $\satdiff$ can merge them into a single deletion. We will present more detailed case studies in Section \ref{sec:case_study1}.

\vspace{-5pt}

\paragraph{Conciseness on OCaml benchmark}
Figure \ref{fig:nate-vs-satdiff} presents the statistics of absolute differences and relative differences in edit script sizes between NATE's dynamic programming-based tree diffing algorithm and $\satdiff$ on the OCaml dataset. We observe that $\satdiff$ uniformly outperforms NATE's tree diffing algorithm by a large margin. More specifically, NATE requires $14.3 \times$ more edits compared to $\satdiff$ on average. For very few cases, we find that $\satdiff$ produces more edits. Upon closer examination, those cases correspond to instances where $\satdiff$ experienced a timeout, thus lacking its minimal guarantees. Our current implementation of $\satdiff$ exhibits a relatively low timeout ratio in both datasets, standing at approximately $0.8\%$. We plan to enhance the timeout ratio in our future work.

\subsection{Evaluation on runtime}

% \paragraph{Runtime on Python benchmark}

% \#(\satdiff better than gumtree): 119

\begin{wrapfigure}[14]{r}{7cm}
    \vspace{-0.5cm}
    % \begin{figure}
        \includegraphics[scale=0.4]{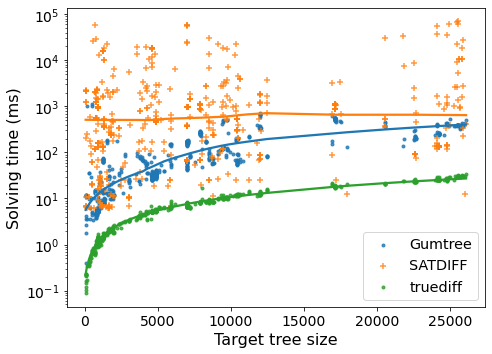}
        \caption{Running time statistics of different approaches.}
        \label{fig:time_treesize}
% \end{figure}
\end{wrapfigure}

We record and compare the runtime performance of $\satdiff$ and baseline approaches using the Python dataset. Figure~\ref{fig:time_treesize} illustrates the runtime of different approaches for each instance, with instances sorted by target tree size. All data points are fitted with regression lines. We observe that \textit{truediff}, which has a linear runtime, outperforms both \textit{Gumtree} and $\satdiff$ by a significant margin.

Moreover, $\satdiff$ runs faster than \textit{Gumtree} in 119 instances, despite the latter demonstrating an overall runtime advantage. It is worth mentioning that the performance gap between $\satdiff$ and \textit{Gumtree} diminishes dramatically as the target tree grows larger. In addition, $\satdiff$ exhibits a much larger variance compared to its counterpart. This is not surprising, as solving {\maxsat} problems has worst-case exponential time. Despite $\satdiff$ not matching the runtime performance of \textit{truediff} and \textit{Gumtree}, it can, on average, finish computing the minimum edit scripts in less than 1 second. We believe this is a reasonable duration from the user's perspective.

\begin{table}[!h]
    \centering
    \begin{tabular}{|l|rrrr|r|}
    \hline
        
    \multirow{2}{4em}{Configuration} & \multicolumn{4}{c|}{Distribution of running time}  & Encoding size \\ 
    & $< 0.1s$ & $ 0.1-1s$ & $1-60s$ & timeout &  Avg. \# clauses\\
        \hline\hline
        \satdiff & 14.22\% & 57.59\% & 20.37\% & 7.82\% & 128,627\\
        \satdiff+opt 1 & 68.48\% & 21.16\% & 5.62\% & 4.74\% & 47,588\\
        \satdiff+opt 2  & 38.10\% & 44.78\% & 10.80\% & 6.32\% & 75,616\\
        \satdiff+opt 1+opt 2& 83.32\% & 8.78\% & 7.02\% & 0.88\% & 34,574\\  
        \hline
    \end{tabular}
    \caption{Results of the ablation study on our optimization techniques. In the first four columns, we report the percentage of tasks that are completed within 0.1s, within 1s, above 1s, and timeout (> 60s) respectively on the OCaml dataset. We also report the average clause number for different configurations in the last column.}
    \label{tab:ablation}
    \vspace{-20pt}
\end{table}

\subsection{Effectiveness of our encoding optimizations.}
\label{sec:opt}

A significant concern with $\satdiff$ is the possibility of an explosion in the number of clauses required to represent the real-world tree diffing problem as a MaxSAT instance. Initially, our implementation involved duplicating all nodes from the target tree into the source space and encoding potential edges comprising all possible pairs of nodes. This resulted in a serious issue of clause explosion. To address this problem, we implemented two key optimizations:
\begin{enumerate}
    \item Computing \emph{virtual nodes} using the multiset set difference $N_T \setminus N_S$ instead of $N_T$.
    \item Encoding \emph{potential edges} selectively by filtering them based on $label$ and $value$.
\end{enumerate}

Table \ref{tab:ablation} presents the results of our ablation study on optimization techniques using the OCaml benchmark. Notably, in the absence of optimization techniques, we can only compute minimum edit scripts for $14.22\%$ of total instances in less than 0.1 seconds. This ratio dramatically increases to $83.32\%$ after applying the proposed encoding optimizations. The ratio of timeouts also decreases significantly from $7.82\%$ to $0.88\%$, along with the average clause number shrinking to nearly one-quarter of its original value, decreasing from 128,627 to 34,574. This suggests the effectiveness of our encoding optimizations. In addition, we find that introducing fewer \emph{virtual nodes} seems to have a more significant improvement effect than filtering \emph{potential edges}. This is understandable as having fewer nodes in the source space reduces both the number of match variables and edge variables, whereas filtering \emph{potential edges} only decreases the number of edge variables.

\subsection{Case study}
\label{sec:case_study1}

% \haolin{\textbf{1.(SAT-Diff > Gumtree)} SAT-Diff catches node updates every time, but there are cases where Gumtree can miss some updates. Ex. \textit{$keras-17-c8f66d15/multi_gpu_utils.py$}}

% \haolin{\textbf{2.(SAT-Diff > Gumtree)} SAT-Diff deletes a subtree with one step, but Gumtree takes number of steps equals to the subtree size. Ex. \textit{$keras-10-4d59675b/metrics_training_test.py$}}

% In our analysis of over 600 samples, we compare the conciseness of edit scripts generated by \textit{Gumtree} and \satdiff, and summaarize into three distinct caterogies: those where \textit{Gumtree} is better, those where \satdiff is better, and those where they have a tie.

% \begin{figure}[h]
\begin{wrapfigure}[17]{r}{7cm}
    \vspace{-1.2cm}
    \centering
    \includegraphics[scale=0.35, trim={0 1.3cm 0cm 0cm}, clip]{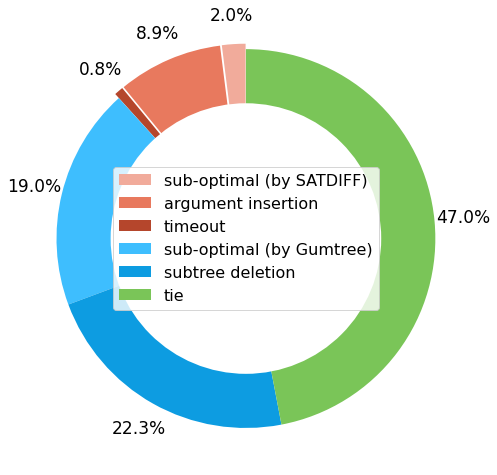}
    \caption{Distribution of the comparison of \textit{Gumtree} and \satdiff. Here red colors represent when \gumtree \ outperforms \satdiff, blue when \satdiff \ outperforms \gumtree. In addition, green refers to the cases where there is a tie.}
    \label{fig:pie}
\end{wrapfigure}

In our analysis of over 600 samples, we compare the conciseness of edit scripts generated by \textit{Gumtree} and \satdiff, and summarize them into three distinct categories: cases where \textit{Gumtree} performs better, cases where {\satdiff} is superior, and cases where they have a tie. To better understand their discrepancies and investigate how \textit{Gumtree} and \satdiff \ could outperform each other, we manually investigate each case and divide them into more refined categories. Figure \ref{fig:pie} presents the distribution of each category in a pie chart.

% We construct a human evaluation on the first two samples, leading to the identification of distinct categories, each corresponding to a specific reason for \textit{Gumtree}'s/\satdiff's superior performance. We calculate the percentages and present them in a pie chart in figure \ref{fig:pie}, with red, blue, and green indicating each of the three sets.

% In most cases, minimum low-level edit scripts are subsequently synthesized into minimum high-level edit scripts. However, it's interesting to note that in specific situations, a minimum low-level edit script may not necessarily translate into a minimum high-level edit script. This can happen because different high-level edit actions may consume a varying number of atomic edit actions. 

% \end{figure} 
As we can see, \textit{Gumtree} generates shorter edit scripts than {\satdiff} in 11.7\% of the total samples, highlighted in red sections. It can further be divided into sub-categories: "sub-optimal (by \satdiff)", "insertion argument", and "timeout". "sub-optimal (by \satdiff)" refers to samples where {\satdiff} deduces sub-optimal high-level edit scripts from \textit{minimum} low-level edit scripts, which accounts for 2.0\% of the samples. 
Specifically, because high-level edit actions consume different numbers of atomic actions, a minimum high-level edit script is sometimes not minimum in its low-level state.
% minimum low-level edit scripts can often be synthesized into minimum high-level edit scripts. However, it's interesting to note that in specific situations, a minimum low-level edit script may not necessarily translate into a minimum high-level edit script.
% This can happen because different high-level edit actions may consume a varying number of atomic edit actions. 
"Argument insertion" signifies instances where {\satdiff} requires more edits than \gumtree \ to insert a child node. Specifically, while \gumtree \ can execute a single insertion in one step, \satdiff \ needs to first create a free slot, which may involve additional operations to shift other children, and then complete the insertion. Although this design choice sometimes increases the edit script size, it makes {\satdiff} more structure-aware and can be extended to support type safety when necessary. The "timeout" category comprises instances in which {\satdiff} fails to find their minimum edit scripts within 60 seconds, accounting for only 0.8\% of the samples.

% may first require additional operations to shift other children to create a free slot.

% due to the strict encoding of child order. 

% Under {\satdiff}'s setup, when inserting a node to be the $k^{th}$ child of a parent, its children with order greater or equal to $k$ must be shifted to their next indices, in order to create a free slot for the incoming child. 

% Child order is a critical feature for many aspects such as ensuring type safety, optimizing potential edges (Remark \ref{rk:opt_on_edges}), and generating dependency graph. While "argument insertion" is the dominating factor for \textit{Gumtree}'s superior performance, by respecting the child order, {\satdiff} is more structure-aware and can be extended to include type information when necessary. 

We further investigate the categories that {\satdiff} perform better, which accounts for 41.3\% of the entire dataset (blue colors). We identify two primary sub-categories: "redundant changes" and "subtree deletion". The "sub-optimal (by \gumtree)" sub-category makes up 19.0\% of the total cases. It represents instances where \textit{Gumtree} fails to identify the optimal matching plans leading to finding minimum edit scripts. Figure \ref{fig: case study} is one such sample we extracted from the dataset. While the difference should just be the extra expression highlighted on line 1080 from the source file on the left in Figure \ref{fig: case study}, {\gumtree} suggests unnecessary deletion and reinsertion of \colorbox{del}{\texttt{TestBackend}} and \colorbox{del}{\texttt{object}}. Upon close inspection, we believe that these two tokens are mismatched due to corner cases in heuristic-based similarity computation, especially when dealing with large trees. Because {\satdiff} uncovers node matching leads to minimum low-level edit scripts, such redundant edits we observe in Figure \ref{fig: case study} never happen.
\begin{figure}[t]
\hspace{0.3cm}
\begin{minipage}{0.4\textwidth}
\vspace{0.4cm}
    \begin{lstlisting}[language=Python,basicstyle=\ttfamily\footnotesize,numbers=left,keywordstyle=\color{mykeywords},
    commentstyle=\color{green!50!black},
    stringstyle=\color{red},    showstringspaces=fals,firstnumber=241,numberstyle={\tiny\color{gray}\ttfamily}]
class (*@\colorbox{del}{TestBackend}@*)((*@\colorbox{del}{object}@*)):
        (*@\hspace*{3em}\footnotesize\dots \setcounter{lstnumber}{1078}@*)
    def test_nn_operations(self):
        (*@\colorbox{del}{check\_single\_tensor\_operation(}@*)       (*@\colorbox{del}{'soft-sign', (4,10), WITH\_NP)}@*)
        check_single_tensor_operation('soft-plus',(4,10),WITH_NP) 
        (*@\hspace*{3em}\footnotesize\dots@*)
    \end{lstlisting}
\end{minipage}\hfill
\hspace{0.5cm}
\vspace{-0.5cm}
\begin{minipage}{0.42\textwidth}
\vspace{-0.5cm}
    \begin{lstlisting}[language=Python,basicstyle=\ttfamily\footnotesize,numbers=left,keywordstyle=\color{mykeywords},commentstyle=\color{green!50!black},
    stringstyle=\color{red},
    showstringspaces=false,firstnumber=241,numberstyle={\tiny\color{gray}\ttfamily}]
class (*@\colorbox{ins}{TestBackend}@*)((*@\colorbox{ins}{object}@*)):  
        (*@\hspace*{3em}\footnotesize\dots \setcounter{lstnumber}{1078}@*)
    def test_nn_operations(self):
        check_single_tensor_operation('soft-plus',(4,10),WITH_NP)
        (*@\hspace*{3em}\footnotesize\dots@*)
    \end{lstlisting}   
\end{minipage}
\caption{An example (keras-96-17082f61/backend\_test.py) where \gumtree \  makes unnecessary code edits.}
\label{fig: case study}
\vspace{-0.5cm}
\end{figure}

The "subtree deletion" category represents a collection of cases in which {\satdiff} produces shorter edit scripts than \gumtree, constituting 22.3\% of the dataset. If we visualize the changes by highlighting them in files, having a convenient deletion mechanism does not distinguish this design choice in $\satdiff$ from that in $\gumtree$. However, when comprehending the edit scripts as sequences of actions, the edit script with "subtree deletion" is more concise and readable for end users. For example, if we have an edit script that involves two subtree deletions, \satdiff \ allows the user to efficiently identify the roots of the removed subtrees. On the other hand, 'subtree deletion' can be valuable in computing similarity scores for specific tasks, such as plagiarism detection.

\section{Related work}
\label{sec:related}

Tree diffing approaches operating at the AST level have a significant advantage over Textual-level differencing techniques, such as Unix {\fontfamily{cmss}\selectfont diff}, as they generate edit scripts that incorporate the tree structure of the code. This enables a more comprehensive analysis of code differences, encompassing a wider range of edit actions and providing a deeper understanding of code evolution between versions. One pioneering work in this line of research was proposed by \cite{ChawatheRGW96}. The approach begins by identifying matching nodes between the source and target trees through heuristics that assess their similarity. Once the matching nodes are established, an optimal edit script is derived, comprising operations such as insertions, deletions, moves, and updates. The original heuristics proposed by Chawathe et al. are specifically optimized for comparing flatly structured documents like LaTeX files. Since then, most tree-diffing approaches have adhered to the same paradigm but employed different similarity metrics for node matching. 

% \vspace{-5 pt}

ChangeDistiller \cite{ChangeDistilling} utilizes a bigram string similarity measure and inner node similarity weighting to improve node matching on fine-grained syntax trees. In contrast, GumTree \cite{gumtree} employs a top-down greedy search to identify the largest isomorphic subtrees between the source and target trees, enabling fine-grained differencing of Java ASTs. Additionally, JSync \cite{JSync} introduces characteristic vectors for measuring structural similarity between (sub)trees and uses heuristics to identify highly similar pairs as clones, facilitating clone detection in JavaScript. Furthermore, X-Diff \cite{XDiff} designs node signatures as the primary criterion for node matching and presents an effective algorithm that combines XML structure characteristics with standard tree diffing techniques. MTDIFF \cite{MTDIFF} incorporates several optimizations for identifying unchanged code, mapping inner nodes, and other processes. These optimizations result in improved accuracy when detecting code movements. 

Additionally, there is a line of work \cite{MTDIFF,hdiff,truediff} that adopts a different strategy for searching node matching, leveraging hashmaps instead of the similarity scores paradigm as seen in Chawathe et al. More specifically, they compute a unique cryptographic hash for each subtree, and matches are found if and only if two subtrees' hashes are equal. While hashmap methods are more efficient compared to similarity-based methods, which often have a quadratic running time, it is important to note that finding mapping pairs based on the same hash value imposes a more restricted condition. This condition may underestimate the potential matched nodes, thereby limiting the effectiveness of the approach. In contrast, our proposed approach, {\satdiff}, is novel as it does not rely on heuristics. The encoding of {\satdiff} allows exhaustive search on node matching, making it capable of providing an optimal edit script compared to heuristics-based approaches.

\section{Conclusion}
\label{sec:conclusion}

We propose a novel and practical approach, {\satdiff}, for addressing the tree diffing problem. In contrast to existing heuristics-based approaches, {\satdiff} reduces the problem to a {\maxsat} problem and leverages state-of-the-art solvers to find the optimal solution. We then decode the optimal solution into minimum low-level tree edits, which are subsequently used to synthesize high-level edit scripts. {\satdiff} enables type-safe editing, a feature not supported in previous approaches like \textit{Gumtree}. Additionally, we formalize the semantics of edit scripts and prove the correctness and minimality of {\satdiff}. We evaluate {\satdiff} by comparing it with state-of-the-art tree diffing algorithms on real-world datasets. Our experimental results suggest that {\satdiff} outperforms existing approaches significantly in terms of conciseness while maintaining reasonable running time performance.

In the future, we aim to optimize {\satdiff} by leveraging Unix {\fontfamily{cmss}\selectfont diff} to quickly figure out partial code matching. Furthermore, extending {\satdiff} for syntactically-broken programs is another interesting direction. To make a broader real-world impact, we plan to extend {\satdiff} to other languages and develop a plugin based on {\satdiff} for popular IDEs like Visual Studio Code.

\newpage
\bibliographystyle{ACM-Reference-Format}
\bibliography{main}


\pagebreak

\appendix

% \section{Algorithm}

% \input{algo3}

\section{Type system and type safety of our generated edit scripts}
\label{sec:type}
Inspired by the type system employed in \textit{truediff} \cite{truediff}, we introduce a linear type system \cite{DBLP:conf/ifip2/Wadler90, walker2005substructural} to track and guarantee the type safety of our edit scripts.

\begin{definition}[The linear type system of edit scripts]
We define the type system as a typing relation written:
$\sum \vdash \delta: (N_{free}^\mathcal{T} \bullet E_{free}^\mathcal{T}) \rhd (N_{free}^{\mathcal{T}'} \bullet E_{free}^{\mathcal{T}'})$, where:
\begin{itemize}
    \item $\delta$ is an edit operation, representing the minimum element of an edit script;
    \item $\sum$ are the signatures of node $n$, defined by:
    \begin{align*}
    &\sum ::= \varepsilon \mid \sum, n:sig \\ &sig ::=  \bigl \langle e_1:\mathcal{T}_1,\dots,e_m:\mathcal{T}_m \bigr \rangle \rightarrow \mathcal{T}    
    \end{align*}
    where each $e_i$ is an edge to a subtree of type $\mathcal{T}_i$, and $\mathcal{T}$ is the type of node $n$;
    \item $N_{free}^\mathcal{T}$  are the $free$ nodes with their type, defined as $N_{free}^\mathcal{T} ::= \varepsilon \mid N_{free}^\mathcal{T}, n : \mathcal{T})$;
    \item $E_{free}^\mathcal{T}$ are edges in $free$ state with their type, defined as $(E_{free}^\mathcal{T} ::= \varepsilon  \mid E_{free}^\mathcal{T}, \fe{n}{i} : \mathcal{T} )$

\end{itemize}
\end{definition}

\begin{lemma}
\label{lemma:type_low}
[\emph{Typing rules of low-level edit script $\underline{\Delta}$}]

\begin{prooftree}
\AxiomC{$n \notin N_{free}$}
\AxiomC{$ (n_p,e_i)\notin E_{free}$}
\AxiomC{$\Sigma(n) = \bigl \langle\dots \bigr \rangle \rightarrow \mathcal{T}$}
\AxiomC{$\Sigma(n_p) =  \bigl \langle \dots ,e_i:\mathcal{T}_i,\dots \bigr \rangle  \rightarrow \mathcal{T}'  $}
\LeftLabel{T-\texttt{dcon}}
\QuaternaryInfC{$\Sigma \vdash 
\dcon{n_p}{i}{n}: (N_{free}^\mathcal{T} \bullet E_{free}^\mathcal{T}) \rhd (N_{free}^\mathcal{T},n: \mathcal{T} \bullet E_{free}^\mathcal{T}, (n_p,e_i): \mathcal{T}_i)$}
\end{prooftree}

\begin{prooftree}
\AxiomC{$n \in N_{free}$}
\AxiomC{$ (n_p,e_i) \in E_{free}$}
\AxiomC{$\mathcal{T} <: \mathcal{T}'$}
\LeftLabel{T-\texttt{con}}
\TrinaryInfC{$\sum \vdash 
\con{n_p}{i}{n}$ : 
$(N_{free}^\mathcal{T},n: \mathcal{T} \bullet E_{free}^\mathcal{T}, (n_p,e_i): \mathcal{T}') \rhd (N_{free}^\mathcal{T} \bullet E_{free}^\mathcal{T})$}
\end{prooftree}

% \begin{prooftree}
%     \AxiomC{$\sum (n) =  \bigl \langle x_1: T_1, \dots, x_m: T_m\bigr \rangle \rightarrow T' $}
%     \AxiomC{${k_1,\dots, k_m} \cap dom(R) = \emptyset $}
%     \LeftLabel{T-\texttt{del}}
%     \BinaryInfC{$\sum \vdash \del{n_p}{i}{n}: (\accentset{\circ}{N_T}, n : T \bullet \accentset{\circ}{E_T}) \rhd (\accentset{\circ}{N_T},k_1 : T_1,\dots,k_m : T_m \bullet \accentset{\circ}{E_T})$}
% \end{prooftree}

\begin{prooftree}
\AxiomC{$n \notin N_{free}$}
\AxiomC{$ (n_p,e_i)\notin E_{free}$}
\AxiomC{$\Sigma(n_p) =  \bigl \langle \dots ,e_i:\mathcal{T}_i,\dots \bigr \rangle  \rightarrow \mathcal{T}'  $}
\LeftLabel{T-\texttt{del}}
\TrinaryInfC{$\Sigma \vdash 
\del{n_p}{i}{n}: (N_{free}^\mathcal{T} \bullet E_{free}^\mathcal{T}) \rhd (N_{free}^\mathcal{T} \bullet E_{free}^\mathcal{T}, (n_p,e_i): \mathcal{T}_i)$}
\end{prooftree}

\begin{prooftree}
\AxiomC{}    \LeftLabel{T-$\underline{\Delta}$-Null}
\alwaysSingleLine
\UnaryInfC{$\sum \vdash \varnothing$: $(N_{free}^\mathcal{T}  \bullet E_{free}^\mathcal{T} ) \rhd (N_{free}^\mathcal{T} \bullet E_{free}^\mathcal{T} )$}
\end{prooftree}

\begin{prooftree}
\AxiomC{$\sum \vdash \delta : (N_{free}^\mathcal{T} \bullet E_{free}^\mathcal{T}) \rhd (N_{free}^{\mathcal{T}'} \bullet E_{free}^{\mathcal{T}'}) $}
\AxiomC{$\sum \vdash \Delta$: $(N_{free}^{\mathcal{T}'} \bullet E_{free}^{\mathcal{T}'} ) \rhd (N_{free}^{\mathcal{T}''}\bullet E_{free}^{\mathcal{T}''})$}
\LeftLabel{T-$\underline{\Delta}$-Cons}
\BinaryInfC{$\sum \vdash (\delta,\Delta)$: $( N_{free}^{\mathcal{T}} \bullet E_{free}^{\mathcal{T}}) \rhd (N_{free}^{\mathcal{T}''} \bullet E_{free}^{T''})$}
\end{prooftree}
% \begin{prooftree}
%     \AxiomC{}
%    \LeftLabel{$\mathcal{T}$-$\underline{\Delta}$-Null}
%     \alwaysSingleLine
%     \UnaryInfC{$\sum \vdash \varnothing$: $(N_{free}^\mathcal{T}  \bullet E_{free}^\mathcal{T} ) \rhd (N_{free}^\mathcal{T} \bullet E_{free}^\mathcal{T} )$}
% \end{prooftree}

% \begin{prooftree}
%     \AxiomC{$\sum \vdash \delta : (N_{free}^\mathcal{T} \bullet E_{free}^\mathcal{T}) \rhd (N_{free}^{\mathcal{T}'} \bullet E_{free}^{\mathcal{T}'}) $}
%     \noLine
%     \UnaryInfC{$\sum \vdash \Delta$: $(N_{free}^{\mathcal{T}'} \bullet E_{free}^{\mathcal{T}'} ) \rhd (N_{free}^{\mathcal{T}''}\bullet E_{free}^{\mathcal{T}''})$}
%     \LeftLabel{$\mathcal{T}$-$\underline{\Delta}$-Cons}
%     \UnaryInfC{$\sum \vdash (\delta,\Delta)$: $( N_{free}^{\mathcal{T}} \bullet E_{free}^{\mathcal{T}}) \rhd (N_{free}^{\mathcal{T}''} \bullet E_{free}^{T''})$}
% \end{prooftree}

% where $ \rho := \langle N, F, E \rangle, $ $ N,E \text{ are node and edge sets of current source space } T_G$ respectively, $F$ is set of \emph{free} nodes and edges in the current source space $T_G$, i.e., $F:=\accentset{\circ}{N} \cup \accentset{\circ}{E}$. $S_1$ and $S_2$ are any two sequences of low-level edit scripts.
\end{lemma}

\begin{definition}[Type safe edit script] Given a source tree $S'$ and a target tree $T$,  we say an edit script is type safe if: 1)$\sum \vdash \Delta: (N_T \bullet \epsilon) \rhd (\epsilon \bullet \epsilon)$, i.e., there are no $free$ edges and nodes left after applying $\Delta$ to $T$; 2) Each edit action strictly follows type rules.
\end{definition}

\begin{theorem}[$\underline{\Delta}$ is type safe]  
\label{thm:type-safe-low-edit} Any correct low-level edit script $\underline{\Delta}$ is type safe.
\end{theorem}
\begin{proof}
We present a proof sketch demonstrating the type safety of $\underline{\Delta}$. Since $\underline{\Delta}$ is correct, it indicates that there are no $free$ edges and nodes left after executing $\Delta$. Moreover, $\underline{\Delta}$ exhibits an appropriate execution order, i.e., an edge or node is used only when it is in the $free$ state. Thus, each action follows type rules, ensuring the type safety of the low-level edit script.
\end{proof}

\begin{lemma}
\label{lemma:type_low}
[\emph{Typing rules of high-level edit script $\overline{\Delta}$}]
\vspace{10 pt}
\begin{prooftree}
    \AxiomC{$n \notin N_{free}$}
    \AxiomC{$n' \notin N_{free}$}
    \AxiomC{$(n_p,e_i) \notin E_{free}$}
    \AxiomC{$(n_p',e_j) \notin E_{free}$}
    \LeftLabel{T-\texttt{Swp}}
    \QuaternaryInfC{$\sum \vdash \Swp{n}{n_p}{i}{n'}{n_p'}{j} : (N_{free}^\mathcal{T} \bullet E_{free}^\mathcal{T})\rhd (N_{free}^\mathcal{T} \bullet E_{free}^\mathcal{T}) $ }
\end{prooftree}

\begin{prooftree}
    \AxiomC{$n \notin N_{free}$}
    \AxiomC{$ (n_p,e_i)\notin E_{free}$}
    \AxiomC{$\Sigma(n_p) =  \bigl \langle \dots ,e_i:\mathcal{T}_i,\dots \bigr \rangle  \rightarrow \mathcal{T}'  $}
    \LeftLabel{T-\texttt{Del}}
    \TrinaryInfC{$\sum \vdash \Del{n_p}{i}{n}:(N_{free}^\mathcal{T} \bullet E_{free}^\mathcal{T}) \rhd (N_{free}^\mathcal{T} \bullet E_{free}^\mathcal{T}), (n_p,e_i): \mathcal{T}_i)$}
\end{prooftree}

\begin{prooftree}
    \AxiomC{$n' \in N_{free}$}
    \AxiomC{$\Sigma(n.par) =  \bigl \langle \dots ,e_i:\mathcal{T}_i,\dots \bigr \rangle  \rightarrow \mathcal{T}' $}
    \AxiomC{$\mathcal{T} <: \mathcal{T}_i$}
    \LeftLabel{T-\texttt{Upd}}
    \TrinaryInfC{$\sum \vdash \Upd{n}{l}{v}: (N_{free}^\mathcal{T}, n':\mathcal{T} \bullet E_{free}^\mathcal{T})\rhd (N_{free}^\mathcal{T} \bullet E_{free}^\mathcal{T})$}
\end{prooftree}

\begin{prooftree}
    \AxiomC{$n \in N_{free}$}
    \AxiomC{$ (n_p,e_i) \in E_{free}$}
    \AxiomC{$\mathcal{T} <: \mathcal{T}'$}
    \LeftLabel{T-\texttt{Ins}}
    \TrinaryInfC{$\sum \vdash \texttt{Insert}(n_p,e_i,n): (N_{free}^\mathcal{T},n: \mathcal{T} \bullet E_{free}^\mathcal{T}, (n_p,e_i): \mathcal{T}') \rhd (N_{free}^\mathcal{T} \bullet E_{free}^\mathcal{T})$}
\end{prooftree}

\begin{prooftree}
    \AxiomC{$ (n_p',e_j) \in E_{free}$}
    \AxiomC{$\Sigma(n) = \bigl \langle\dots \bigr \rangle \rightarrow \mathcal{T}$}
    \AxiomC{$\mathcal{T} <: \mathcal{T}'$}
\AxiomC{$\Sigma(n_p) =  \bigl \langle \dots ,e_i:\mathcal{T}_i,\dots \bigr \rangle  \rightarrow \mathcal{T}'' $}
    \LeftLabel{T-\texttt{Mov}}
    \QuaternaryInfC{$\sum \vdash 
    \Mov{n}{n_p}{i}{n_p'}{j}: (N_{free}^\mathcal{T} \bullet E_{free}^\mathcal{T}, (n_p',e_j): \mathcal{T}') \rhd (N_{free}^\mathcal{T}  \bullet E_{free}^\mathcal{T}, (n_p,e_i): \mathcal{T}_i) $ }
\end{prooftree}

\hspace{-0.5cm}
% \vspace{1cm}
\begin{minipage}{0.47\textwidth}
    \begin{prooftree}
    \AxiomC{}    
    \LeftLabel{T-$\underline{\Delta}$-Null}
    \alwaysSingleLine
    \UnaryInfC{$\sum \vdash \varnothing$: $(N_{free}^\mathcal{T}  \bullet E_{free}^\mathcal{T} ) \rhd (N_{free}^\mathcal{T} \bullet E_{free}^\mathcal{T} )$}
    \end{prooftree}
\end{minipage}\hfill
\begin{minipage}{0.47\textwidth}
\begin{prooftree}
    \AxiomC{$\sum \vdash \delta : (N_{free}^\mathcal{T} \bullet E_{free}^\mathcal{T}) \rhd (N_{free}^{\mathcal{T}'} \bullet E_{free}^{\mathcal{T}'}) $}
    \noLine
    \UnaryInfC{$\sum \vdash \Delta$: $(N_{free}^{\mathcal{T}'} \bullet E_{free}^{\mathcal{T}'} ) \rhd (N_{free}^{\mathcal{T}''}\bullet E_{free}^{\mathcal{T}''})$}
    \LeftLabel{T-$\underline{\Delta}$-Cons}
    \UnaryInfC{$\sum \vdash (\delta,\Delta)$: $( N_{free}^{\mathcal{T}} \bullet E_{free}^{\mathcal{T}}) \rhd (N_{free}^{\mathcal{T}''} \bullet E_{free}^{\mathcal{T}''})$}
\end{prooftree} 
\end{minipage}

% \begin{prooftree}
%     \AxiomC{}
%     \LeftLabel{T-$\overline{\Delta}$-Null}
%     \alwaysSingleLine
%     \UnaryInfC{$\sum \vdash \varnothing$: $(N_{free}^\mathcal{T} \bullet E_{free}^\mathcal{T}) \rhd (N_{free}^\mathcal{T} \bullet E_{free}^\mathcal{T})$}
%     \DisplayProof \ \ \ \ \
%     \AxiomC{$\sum \vdash \delta : (N_{free}^\mathcal{T} \bullet E_{free}^\mathcal{T}) \rhd (N_{free}^\mathcal{T'} \bullet E_{free}^\mathcal{T'}) $}
%     \noLine
%     \UnaryInfC{$\sum \vdash \Delta$: $(\accentset{\circ}{N_T}'\bullet \accentset{\circ}{E_T}') \rhd (\accentset{\circ}{N_T}'' \bullet \accentset{\circ}{E_T}'')$}
%     \LeftLabel{T-$\overline{\Delta}$-Cons}
%     \UnaryInfC{$\sum \vdash (\delta,\Delta)$: $(\accentset{\circ}{N_T} \bullet \accentset{\circ}{E_T}) \rhd (\accentset{\circ}{N_T}'' \bullet \accentset{\circ}{E_T}'')$}
% \end{prooftree}
\end{lemma}

\begin{theorem}[$\overline{\Delta}$ is type safe]  
\label{thm:type-safe-low-edit} The $\overline{\Delta}$, synthesized from a correct $\underline{\Delta}$ using dependency graph $K$, is also type safe.
\end{theorem}
\begin{proof}
We present a proof sketch. Note that the synthesis process only merges certain subsequences of actions into high-level edit actions without changing the appropriate execution order of $\underline{\Delta}$. Thus$\overline{\Delta}$ is type safe given that $\underline{\Delta}$ is type safe .

% Given that $K_{\overline{\Delta}}$ is correct, if $\underline{\Delta}$ is correct, then $\overline{\Delta}$ must also be correct. This is because $\overline{\Delta}$ is essentially a permutation of $\underline{\Delta}$ with the same edit effect. 
% Consequently, if $\underline{\Delta}$ is type-safe, then $\overline{\Delta}$ is also type safe. This is because $K_{\overline{\Delta}}$ adheres to the type rules of high-level operations, and no low-level actions are left that could compromise type safety.
\end{proof}

\end{document}